\documentclass{article}

\usepackage{parskip}
\usepackage[margin=1in]{geometry}

\usepackage[utf8]{inputenc} 
\usepackage[T1]{fontenc}
\usepackage{hyperref}
\usepackage[cmex10]{mathtools}
\usepackage{amsthm}
\usepackage{enumitem}
\usepackage{microtype}
\usepackage{natbib}

\mathtoolsset{showonlyrefs}

\usepackage[sfdist]{rrmacros}

\newcommand{\eps}{\epsilon}

\newcommand{\N}{\bbN}
\newcommand{\R}{\bbR}
\newcommand{\psd}{\bbS_+}
\newcommand{\pd}{\bbS_{++}}
\newcommand{\id}{\mathrm{I}}

\DeclareMathOperator{\vol}{vol}
\newdist{\normal}{N}
\newcommand{\gaus}{\mathfrak{g}}
\newcommand{\bmeas}{P}

\newdist{\unif}{Unif}
\newcommand{\sym}{\mathrm{Sym}}
\newcommand{\blank}{{\nonscript\,\cdot\nonscript\,}}
\newcommand{\wA}{\widehat{A}}

\newtheorem{thm}{Theorem}
\newtheorem{prop}{Proposition}
\newtheorem{lem}{Lemma}

\theoremstyle{definition}

\newtheorem{condition}{Condition}
\newtheorem{definition}{Definition}
\newtheorem{remark}{Remark}

\title{Statistical Limits for Finite-Rank Tensor Estimation}

\author{Riccardo Rossetti\textsuperscript{*} and Galen Reeves\textsuperscript{*\textdagger}} \date{\vspace{-.4cm}\small{\textsuperscript{*}Department of Statistical Science, Duke University \\ \textsuperscript{\textdagger}Department of Electrical and Computer Engineering, Duke University}}

\begin{document}
\maketitle

\begin{abstract} 
This paper provides a unified framework for analyzing tensor estimation problems that allow for nonlinear observations, heteroskedastic noise, and covariate information. We study a general class of high-dimensional models where each observation depends on the interactions among a finite number of unknown parameters. Our main results provide asymptotically exact formulas for the mutual information (equivalently, the free energy) as well as the minimum mean-squared error in the Bayes-optimal setting. We then apply this framework to derive sharp characterizations of statistical thresholds for two novel scenarios: (1) tensor estimation in heteroskedastic noise that is independent but not identically distributed, and (2) higher-order assignment problems, where the goal is to recover an unknown permutation from tensor-valued observations. 
\end{abstract}

\tableofcontents

\section{Introduction}
The spiked tensor model, introduced by Montanari and Richard \cite{montanari:2014a-statistical}, is a model for higher-order principal component analysis where tensor-valued observations have the form  
\begin{align}
    y_{\alpha}  = \sqrt{\frac{t}{ n^{q-1}} } \theta_{\alpha_1} \cdots \theta_{\alpha_q}  + z_{\alpha} , \quad \alpha \in [n]^q \label{eq:stm}
\end{align}
for unknown real parameters $\theta_1, \dots, \theta_n$ and i.i.d.\ standard Gaussian noise $(z_\alpha)_{\alpha \in [n]^k}$ . The parameter $t \ge 0$ quantifies the signal-to-noise ratio (SNR) and the case $q=2$ is the Sherrington-Kirkpatrick model~\cite{sherrington:1975solvable} from statistical physics. Over the past decade, this model has played a central role in the theoretical analysis of statistical and computational  tradeoffs for high-dimensional inference problems as well as closely related combinatorial optimization problems~\cite{montanari:2014a-statistical,lesieur:2017statistical,kadmon:2018statistical,ben-arous:2019the-landscape,perry:2020statistical}. The  thresholds for detection and estimation are well understood for settings the parameters are uniformly distributed on the sphere $\{ \bmtheta \in \bbR^n  \mid \| \bmtheta\| \le \sqrt{n}\}$ or the hypercube $\{-1,+1\}^n$~\cite{montanari:2014a-statistical,montanari:2015on-the-limitation,lesieur:2017statistical}. 

Despite its popularity on the theoretical side, however, the spiked tensor model \eqref{eq:stm} is limited in the sense that it cannot represent many of the complex behaviors that occur for tensor estimation problems in  statistics and machine learning: often, the underlying parameters are multidimensional, the observations are asymmetric, heterogeneous, possibly non-linear, and there may be additional covariate information about both the parameters and their interactions. Each of these considerations can have substantial impact on the statistical and computational thresholds of these problems. 

\paragraph{Our contribution.}
We propose a general model for $q$-wise interactions and prove asymptotically exact formulas characterizing its statistical limits. In this model, parameters $\theta_1, \dots, \theta_n$ taking values in generic parameter space $\Theta$ are observed via 
\begin{align}
    y_{\alpha}  = \sqrt{\frac{t}{ n^{q-1}} } f( x_{\alpha_1}, \dots, x_{\alpha_q}, \theta_{\alpha_1}, \dots,  \theta_{\alpha_q})  + z_{\alpha}, \qquad \alpha \in [n]^q 
    \label{eq:qim}
\end{align}
where $x_1, \dots, x_n$ are observed covariates belonging to set $\cX$,  $f \colon \cX^q \times \Theta^q \to \bbR^m$ is a known function, and $z_\alpha \sim_\iid \normal(0,\id_m)$ are i.i.d.\ noise vectors independent of everything else. 

The goal of this paper is to provide a theoretical framework that unifies the analysis of this broad class of models, and at the same time allows for arbitrary orders of interaction.

We study the high-dimensional limit where $(m,q, \cX, \Theta, f)$ are fixed  as $n \to \infty$. 
The main object of interest in the \emph{free energy} associated with a sequence of prior distributions $\bmeas^n$ on the $\cX^n \times \Theta^n$ satisfying certain regularity conditions. 
Following well-established techniques in the literature~\cite{iba:1999the-nishimori,nishimori:2001statistical,mezard:2009information}, formulas for the asymptotic free energy can be combined with perturbation arguments (e.g., small changes in the SNR parameter $t$) to identify the Bayes risk under squared error loss, i.e., the minimum mean-squared error. 

Our approach, which is described in Section~\ref{s:factor_model},  consists of two basic steps. 
\begin{itemize}
\item \textbf{Multilinear approximation:} we introduce a multilinear approximation to \eqref{eq:qim} that depends on a \emph{finite} representation of the covariates $\psi \colon \cX \to [k]$ and a  \emph{finite dimensional} representation of the parameters $\phi \colon \Theta \to \bbR^d$. Theorem~\ref{th:approx} establishes continuity of the free energy associated with these approximations with respect to the quadratic Wasserstein distance between statistical models. The proof, which uses the HWI inequality of Otto and Villani~\cite{otto:2000generalization}, both simplifies and strengthens related continuity results in the literature.  

\item \textbf{Free energy:} we then prove a formula for the asymptotic free energy associated with the multilinear approximation (Theorem~\ref{th:Flimit}). The key idea of the proof is to embed the model into a lifted version of the matrix tensor product model (MTP)~\cite{reeves2020,chen2022} and then use orthogonality properties of the embedding to map back down to a simplified formula defined as max-min of a potential function on the finite dimensional space $(\psd^d)^k \times (\psd^d)^k$. 
\end{itemize}

To illustrate how our approach can be applied to settings that go  beyond the spiked tensor model in \eqref{eq:stm}, we use it to study the statistical limits for two novel problem settings: 
\begin{itemize}
\item  \textbf{Heteroskedastic noise:} We study a version of \eqref{eq:stm} where the Gaussian noise is independent but not identically distributed. Our results generalize prior heteroskedastic matrix models~\cite{behne2022,guionnet:2025low-rank} by allowing for higher-order interaction and removing the assumption that the variance profile is positive definite. 

\item \textbf{Higher-order assignment problems:} The $q$-adic assignment problem introduced by Lawler \cite{lawler:1963the-quadratic} is an extension of the classical quadratic assignment problem ($q=2$) in \cite{koopmans1957} where the goal is to find a permutation of the set $[n]$ that minimizes the discrepancy between two order-$q$ tensors. Similar to recent work on the quadratic setting~\cite{yang2024}, we consider a variation of \eqref{eq:qim} where the parameters represent an unknown permutation. We characterize the limit via a converging sequence of approximations. 
\end{itemize}
We anticipate that our framework can also be applied to other settings involving multi-way comparisons such as multidimensional assignment problems~\cite{pierskalla:1968the-multidimensional,bandelt:1994approximation} and hypergraph clustering~\cite{agarwal:2005beyond,zhou:2006learning}.

\subsection{Related work}
\paragraph{Low-rank matrix models.} Matrix-valued models $(q=2)$ have been studied extensively for applications including  covariance estimation and principal component analysis~\cite{johnstone2001,baik:2005phase,benaych-georges:2011the-eigenvalues,deshpande2014,dia2016,miolane2017,lelarge2019}, as well as community detection with stochastic block models~\cite{abbe:2015recovering,banks:2016information-theoretic,caltagirone:2017recovering,deshpande:2017,abbe:2018}. There has also been progress on richer classes of models allowing for approximately low-rank \cite{lu2022,yang2024} interactions, different types of observations and covariate information \cite{deshpande2018,mayya2019,yang2024}, as well as  heteroskedastic noise  \cite{behne2022,guionnet:2025low-rank,mergny2024,guionnet2024}. 

\paragraph{Theoretical techniques.}
Our work  builds upon and unifies a number of ideas that have appeared previously in the literature. The proof of Theorem~\ref{th:Flimit} depends on  results for the MTP model~\cite{reeves2020,chen2022}, and a simplified version of the lifting argument we employ in this paper was used in \cite{behne2022} for rank-one matrices. Prior work focusing directly on spiked matrix models corresponding to special cases of~\eqref{eq:qim} includes \cite{guionnet:2025low-rank,mergny2024,guionnet2024,yang2024}. Unlike our approach, many of these works restrict the parameter interactions to be positive semidefinite. We also highlight that some of our ideas for approximations to non-linear models based on series expansion have been studied extensively in the general context of low-degree polynomial approximations~\cite{hopkins:2017the-power, hopkins:2017efficient, kunisky:2022notes}.

\paragraph{Computational thresholds.} It is also important to know if and when these statistical limits can be attained with  practical computational complexity. In the $q=2$ setting, there is often a close connection between the free energy formulas and the performance of approximate message passing (AMP) algorithms, pointing to the fact that for the matrix case the entire regime in which estimation is possible is achievable efficiently. However, for the tensor setting with $q > 2$, even AMP-based methods fail to achieve the fundamental limits. While more powerful algorithms have been devised, there remain significant statistical-to-computational gaps~\cite{montanari:2014a-statistical,lesieur:2017statistical,perry:2020statistical,ben-arous:2019the-landscape,ben-arous:2020algorithmic}.
While this paper focuses on the statistical limits, we think the computational limits are an important and complementary direction for future work. 

\paragraph{Tensor estimation methods.} A wide range of techniques has been proposed to estimate the signal $\theta_1,\dots,\theta_n$ from a spiked tensor model~\eqref{eq:stm}. Perhaps the most classical method is the tensor power iteration, a direct analog of the standard power iteration on symmetric matrices whose statistical properties are well-understood~\cite{montanari:2014a-statistical,huang:2022power,wu:2024sharp}. Other approaches with the same statistical performance include gradient descent~\cite{ben-arous:2020algorithmic} and AMP~\cite{montanari:2014a-statistical}. In further developments, it was shown that the power iteration is suboptimal in the sense that there are some signal regimes in which other computationally efficient methods achieve recovery. Among these, we mention sum-of-squares methods~\cite{hopkins:2015tensor,hopkins:2016fast}, higher-order orthogonal iterations~\cite{luo:2021a-sharp}, and modified gradient methods~\cite{anandkumar:2014guaranteed,biroli:2020how-to-iron,han:2022an-optimal}. Development of provably efficient algorithms for more complex interaction structures such as~\eqref{eq:qim} is, to the best of our knowledge, an open problem.

\section{The \texorpdfstring{$q$-wise}{q-wise} interaction model}\label{s:factor_model}
In this paper, we will be considering statistical models consisting of array-valued random variables, which we will be indexing by vectors of integers. We first present some notation that will be used throughout.
\begin{itemize}[leftmargin=2em]
    \item For an integer $n \in \bbN$, we define $[n] \coloneqq \{1,\dots,n\}$;
    \item for an $n$-tuple of elements $\bmv \coloneqq (v_1,\dots,v_n)$ and $\alpha \coloneqq (\alpha_1,\dots,\alpha_q) \in [n]^{q}$, we write $\bmv(\alpha)$ to denote the $q$-tuple of elements $(v_{\alpha_1},\dots,v_{\alpha_q})$;
    \item we use $\inner{}{}$ for the Euclidean inner product, and denote by $\norm{}$ the induced norm. Given two elements in a product space $v,w \in (\bbR^m)^n$, we use the convention $\inner{v}{w} \coloneqq \sum_{i=1}^n \inner{v_i}{w_i}$.
\end{itemize}

Fix two integers $q,m \in \bbN$. For each problem size $n \in \bbN$, let $N \coloneqq n^q$. We define the general \textit{$q$-wise interaction model} as an array of $\bbR^m$-valued observations $\bmy \coloneqq (y_\alpha)_{\alpha \in [n]^q}$ given by
\begin{align}
    y_\alpha = \sqrt{\frac{t}{n^{q-1}}} f(x_{\alpha_1},\dots,x_{\alpha_q}, \theta_{\alpha_1},\dots,\theta_{\alpha_q}) + z_\alpha, \quad \alpha \in [n]^q \label{eq:y_alpha}
\end{align}
where
\begin{itemize}[leftmargin=2em]
    \item $\bmx\in \cX^n$ are observed covariates;
    \item $\bmtheta \in \Theta^n$, are unknown parameters;
    \item $f  : \cX^q  \times \Theta^q \to \bbR^m$ is a known function;
    \item $z_\alpha \in \bbR^m$, $\alpha \in  [n]^q$, are i.i.d. standard Gaussian vectors.
\end{itemize}
To simplify the presentation, for a function $f: \cX^q \times \Theta^q \to \bbR^m$ we define $f^{(n)} : \cX^n \times \Theta^n \to (\bbR^m)^N$ to be the function taking in the covariate and parameter vectors and returning the dimension-scaled array of outputs of $f$ for all indices $\alpha \in [n]^q$. In other words, $f^{(n)}(\bmx,\bmtheta)$ is defined element-wise as
\begin{align}\label{eq:g}
    f_\alpha^{(n)}(\bmx,\bmtheta) \coloneqq \frac{1}{\sqrt{n^{q-1}}}f(x_{\alpha_1},\dots,x_{\alpha_q}, \theta_{\alpha_1},\dots,\theta_{\alpha_q}), \quad \alpha \in [n]^q.
\end{align}
Model  \eqref{eq:y_alpha} is designed to cover a wide variety of high-dimensional factor models studied in the literature. As a consequence of this, the results in this section provide the key tools necessary to study the statistical limits of a wide range of high-dimensional inference models. The flexibility and generality of this approach is demonstrated in Sections~\ref{s:hetero} and~\ref{s:qap}, in which we are able to provide a solution for two apparently unrelated classes of open problems through the use of similar techniques.

\begin{remark}\label{remark:psd} Much the prior work has focused on the special case of pairwise interactions $(q=2)$ under the additional constraint that $f$ is a positive semidefinite kernel.  
In this setting the formulas for the resulting limits can be stated in terms the maximization of a  potential function that has a ``simple'' form. However, because we do not make this assumption, our general results require a more complicated max-min formula that strictly generalizes the positive semidefinite case. As shown in \cite[Proposition~6]{reeves2020}, the max-min formula can be reduced to a single max formula precisely when the positive semidefinite condition holds. See also \cite[Theorem~1.1]{chen2025a} for a discussion of the case $q > 2$.
\end{remark}

\begin{remark}
Based on universality arguments originating in~\cite{korada2011} and subsequently developed in~\cite{lesieur2015,lu2022,hu2022,pesce2023,han2023} one can extend results under the assumed additive Gaussian noise in \eqref{eq:y_alpha} to suitably regular non-Gaussian models; for example,~\cite[Theorem 2.7]{guionnet:2025low-rank} derive a universality result for $q=2$ in the setting studied in Section~\ref{s:hetero}, under an additional positive semidefinite constraint.
\end{remark}

This section contains our main results on the $q$-wise interaction model, and is organized as follows:
\begin{itemize}[leftmargin=2em]
    \item in Section~\ref{ss:setup}, we introduce the probabilistic setting of our problem, as well as all the relevant information-theoretic quantities for our analysis;
    \item in Section~\ref{ss:approx}, we provide an approximation result that allows to reduce the analysis of the $q$-wise interaction model to certain multilinear approximation model, which is defined there;
    \item Section~\ref{ss:F_limit} contains an asymptotically exact characterization of the free energy for the multilinear model that is used to approximate the $q$-wise interaction model;
    \item finally, Section~\ref{ss:overlap} details the connection between free energy and the minimum mean-squared error in a Bayes-optimal estimation framework. 
\end{itemize}
All proofs for this section are provided in Appendix~\ref{s:proof_main}.

\subsection{Bayes-optimal setting} \label{ss:setup}
We focus on the high-dimensional Bayes-optimal setting where, for each problem of size $n \in \bbN$, the unknown parameters $\bmtheta \coloneqq (\theta_i)_{i\in[n]}$ and covariates $\bmx \coloneqq (x_i)_{i\in[n]}$ have known joint distribution $\bmeas^n$ on $\cX^n \times \Theta^n$. 

Throughout, we will use the symbol $\gaus$ to denote the standard Gaussian measure defined on a space whose dimension is appropriate to the context. 
For the pushforward measure of $\bmeas^n$ under function $f^{(n)}$, we write $f^{(n)}_\#\bmeas^n$. 

With these definitions, let us recall the form of the $q$-wise interaction model $\bmy \coloneqq (y_\alpha)_{\alpha \in [n]^q}$,
\begin{align}
    y_\alpha = \sqrt{t} f_\alpha^{(n)}(\bmx,\bmtheta) + z_{\alpha}, \quad \alpha \in [n]^q
\end{align}
whose distribution is given by the convolution $f^{(n)}_\#\bmeas^n * \gaus$. We define the Hamiltonian for this model $\cH^n_{\bmx,\bmtheta,\bmz} : \Theta^n \times [0,\infty) \to \bbR$ as the random function
\begin{align} \label{eq:hamiltonian}
    \cH^n_{\bmx,\bmtheta,\bmz}(\bmtheta';t) \coloneqq t \big(\inner{f^{(n)}(\bmx,\bmtheta')}{\; f^{(n)}(\bmx,\bmtheta)} - \frac{1}{2} \norm{f^{(n)}(\bmx,\bmtheta')}^2\big) + \sqrt{t} \inner{f^{(n)}(\bmx,\bmtheta')}{\bmz}.
\end{align}
We define the (expected) \textit{free energy} $F_n : [0,\infty) \to \bbR$ of $\bmy$ (equivalently, of $\sqrt{t}f^{(n)}_\#\bmeas^n * \gaus$) as
\begin{align}
    F_n(t) \coloneqq \frac{1}{n} \ex[\Big]{\log \int \exp\!\big\{\cH^n_{\bmx,\bmtheta,\bmz}(\bmtheta';t)\big\} \, \bmeas^n(\dd \bmtheta' \mid \bmx) },
\end{align}
where the expectation is taken with respect to all random quantities $(\bmx,\bmtheta,\bmz)$. Note that the above expression is precisely the relative entropy of the conditional distribution of $\bmy \mid \bmx$ with respect to the Gaussian measure.

For Gaussian noise models, recall that the free energy $F_n$ provides an equivalent characterization of the mutual information via the identity
\begin{align}
    I(u + z; u) + \div{\mu * \gaus}{\gaus} = \bbE\norm{u}^2,
\end{align}
for a pair of random variables $(u,z)\sim \mu \otimes \gaus$ defined on the same Euclidean space and $\div{}{}$ denoting the relative entropy.

With this, we find that the mutual information between $\bmtheta$ and $\bmy$ conditionally on $\bmx$ can be equivalently characterized in terms of the free energy $F_n$. In fact, by the chain rule of mutual information
\begin{align}
    I\big(\bmtheta; \bmy \mid \bmx\big) = I\big(f^{(n)}(\bmx,\bmtheta),\bmtheta; \bmy \mid \bmx\big) = I\big(f^{(n)}(\bmx,\bmtheta); \bmy \mid \bmx\big),
\end{align}
where $I(f^{(n)}(\bmx,\bmtheta) ; \bmy \mid \bmtheta,\bmx) = 0$ and furthermore $I(\bmtheta ; \bmy \mid f^{(n)}(\bmx,\bmtheta),\bmx) = 0$ since, conditionally on $\bmx$, $f^{(n)}(\bmx,\bmtheta)$ is a sufficient statistic for $\bmtheta$. In particular, the normalized conditional mutual information satisfies
\begin{align} \label{eq:I_F_identity}
    \frac{1}{n}I\big(\bmy; \bmtheta \mid \bmx\big) = \frac{t}{n} \ex*{\norm[\big]{\ex{f^{(n)}(\bmx,\bmtheta) \mid \bmx}}^2} - F_n(t).
\end{align}

\subsection{Bound on the approximation error} \label{ss:approx}
As a preliminary, let us recall the definition of the \textit{quadratic Wasserstein distance} between two square-integrable measures $\mu,\nu$ on some Euclidean space $\bbR^n$. We have
\begin{align}
    W_2(\mu,\nu) \coloneqq \left( \inf_{\gamma \in \Gamma(\mu,\nu)} \int \norm{v-w}^2 \, \gamma (\dd v,\dd w) \right)^{1/2}, 
\end{align}
where $\Gamma(\mu,\nu)$ is the set of couplings on $\bbR^n \times \bbR^n$ with marginal distributions $\mu$ and $\nu$.

For a generic function $f : \cX^q \times \Theta^q \to \bbR^m$, our final goal is the calculation of the free energy $F_n(t)$ associated to the $q$-wise interaction model
\begin{align}
    \bmy =\sqrt{t} f^{(n)}(\bmx,\bmtheta) + \bmz.
\end{align} 
In our approach, we will consider approximations to $f$ via a specific class of multilinear functions, which are defined as follows.
\begin{definition}
For $d,k \in\bbN$ let $\cF_{d,k}$ be the class of functions $f \colon \cX^q \times \Theta^q \to \bbR^m $ of the form 
\begin{align} \label{eq:multilinear}
    f(x_1,\dots, x_q, \theta_1, \dots, \theta_q) =A_{\psi(x_1)\dots \psi(x_q)} \big( \phi(\theta_1) \otimes \cdots \otimes \phi(\theta_q)\big) 
\end{align}
for bounded measurable functions $\phi \colon \Theta \to \bbR^d$, $\psi \colon \cX \to [k]$, and matrices $A_\kappa \in \bbR^{m \times d^q}$, $\kappa \in [k]^q$. Define $\cF \coloneqq  \cup_{d,k \in \bbN} \cF_{k,d}$.  
\end{definition}

\begin{condition}\label{cond:eps_approx}
For every $\eps > 0$ and $t \in [0,\infty)$, there exists a multilinear approximation $\hat{f} \in \cF$ such that the sequence of functions $\hat{f}^{(n)}\colon \cX^n \times \Theta^n \to (\bbR^m)^{N}$ satisfies
\begin{align}
\limsup_{n \to \infty} \frac{1}{n} \ex*{W_2^2\Big( \big[\sqrt{t}f^{(n)}_\# \bmeas^n * \gaus\big](\blank \mid \bmx), \,  \big[\sqrt{t} \hat f^{(n)}_\# \bmeas^n * \gaus\big](\blank \mid \bmx)} \leq \eps \label{eq:W2eps}.
\end{align}
\end{condition}

\begin{remark}
The assumption that parameter maps $\phi$ are bounded ensures that $\frac{1}{n} \bbE\norm{\hat{f}^{(n)}(\bmx, \bmtheta)}^2$ is bounded uniformly in $n$. Moreover, because convolution is non-expansive under the Wasserstein distance~\cite[Lemma 5.2]{santambrogio:2015optimal}, the Wasserstein bound in \eqref{eq:W2eps} is implied by the simpler but potentially much more restrictive  condition
\begin{align}
\limsup_{n \to \infty} \frac{1}{n} \int \| f^{(n)}- \hat{f}^{(n)} \|^2  \dd \bmeas^{n} \leq \eps. \label{eq:W2eps_simple}
\end{align}
We note that in some settings, including some in the low-SNR regime where $t=O(1)$, convolution with Gaussian measure is highly contractive~\cite{chen:2021asymptotics} and thus \eqref{eq:W2eps} may hold even if \eqref{eq:W2eps_simple} does not. 
\end{remark}
\begin{thm}\label{th:approx}
Assume that Condition~\ref{cond:eps_approx} holds.  Then, for every $\eps >0$ and $t \in [0,\infty)$ there exists a multilinear approximation  $\hat{f} \in \cF$ such that the free energy $\hat{F}_n$ for $(\sqrt{t} \hat f^{(n)}_\#\bmeas^n * \gaus)$ satisfies 
\begin{align}
    \limsup_{n \to \infty} | F_n - \hat{F}_n| \le \eps
\end{align}
\end{thm}

\subsection{Asymptotic free energy} \label{ss:F_limit}
In this section, we establish the exact limit of the free energy for $q$-wise interaction models whose signal components have the multilinear form given in \eqref{eq:multilinear}. For the remainder of this section, we assume that $f^{(n)} \in \cF$, and fix $d,k \in \bbN$, as well as maps $\psi \colon \cX \to [k]$, $\phi \colon \Theta \to \bbR^d$, and matrices $A_\kappa$, $\kappa \in [k]^{q}$. We consider the family of models parametrized by a scalar $t \in [0,\infty)$ given by
\begin{align} \label{eq:y_approx}
\bmy = \sqrt{t}f^{(n)}(\bmx,\bmtheta) + \bmz.
\end{align}
Before proceeding with the statement of our main result for this section, we define the \textit{linear channel model} associated to~\eqref{eq:y_approx}. Given a collection $\mathring S = (\mathring S_1, \dots, \mathring S_k) \in (\psd^d)^k$, we define an additive Gaussian model of the form:
\begin{align} \label{eq:y_ring}
    \mathring{y}_i = \mathring S_{\psi(x_i)}^{1/2} \phi(\theta_i) + \mathring{z}_i , \quad i \in [n]
\end{align}
where $\mathring{z}_i \sim_\iid \normal(0,\id_d)$ are independent of $(\bmx,\bmtheta,\bmz)$. For this linear channel model, we introduce the Hamiltonian, $\mathring\cH_{\bmx,\bmtheta,\mathring\bmz} : \Theta^n \times (\psd^d)^k \to \bbR$ 
\begin{align} \label{eq:hamiltonian_linear}
&\mathring\cH_{\bmx,\bmtheta,\mathring\bmz}^n(\bmtheta'; \mathring S) \coloneqq \sum_{i=1}^n \!\big( \inner{\mathring S_{\psi(x_i)}^{1/2} \phi(\theta_i')}{\mathring S_{\psi(x_i)}^{1/2} \phi(\theta_i) + \mathring{z}_i} - \frac{1}{2} \norm{\mathring S_{\psi(x_i)}^{1/2} \phi(\theta_i')}^2\big).
\end{align} 
With this, the relative entropy $H_n : (\psd^d)^k \to \bbR$ associated with $\mathring\bmy$ can be written as
\begin{align}\label{eq:entropy_linear}
    H_n(\mathring S) \coloneqq \frac{1}{n} \ex[\Big]{\log \int \exp\!\big\{\mathring\cH^n_{\bmx,\bmtheta,\mathring\bmz}(\bmtheta';\mathring S)\big\} \, \bmeas^n(\dd \bmtheta' \mid \bmx) }.
\end{align}
Note that, in the special case where $(x_i,\theta_i)$ are i.i.d.\ from some known measure $\bmeas$ on $\Theta \times \cX$, $H_n$ can be seen not to depend on $n$ and is given in terms of a single-letter formula. 

Let us consider an augmented observation model $\overline\bmy = (\bmy, \mathring \bmy)$ consisting of and observation $\bmy$ of the form~\eqref{eq:y_approx} and the linear channel observations $\mathring \bmy$ in~\eqref{eq:y_ring}. Let us define the conditional density $\rho^n_{\bmx,\bmtheta,\bmz,\mathring\bmz}(\mathring S,t)$ of $\overline \bmy \mid \bmx$ as a function of $(\mathring S,t) \in (\psd^d)^k \times [0,\infty)$: 
\begin{align}
    \rho^n_{\bmx,\bmtheta,\bmz,\mathring\bmz}(\mathring S,t) = \int \exp\!\big\{\cH^n_{\bmx,\bmtheta,\bmz}(\bmtheta';t) + \mathring\cH^n_{\bmx,\bmtheta,\mathring\bmz}(\bmtheta';\mathring S)\big\} \, \bmeas^n(\dd \bmtheta' \mid \bmx).
\end{align}
\begin{condition} \label{as:Flimit} For any $f^{(n)} \in \cF$, the following assumptions hold.
\begin{enumerate}[leftmargin=3em,label=\ref*{as:Flimit}\alph*.,ref=\ref{as:Flimit}\alph*]
    \item  \label{as:Flimit_entropy_convergence} The linear channel relative entropy $H_n$ converges pointwise to a function $H : (\psd^d)^k \to \bbR$ that is continuously differentiable over positive definite cone $(\pd^d)^k$.
    \item \label{as:Flimit_concentration} Let $\cC$ be any compact subset of $(\psd^d)^k \times [0,\infty)$. Then, the density $\rho^n_{\bmx,\bmtheta,\bmz,\mathring\bmz}$ satisfies 
    \begin{align}
        \lim_{n\to \infty} \frac{1}{n}\ex[\bigg]{\sup_{(S,t') \in \cC} \,\abs[\Big]{ \log\rho_{\bmx,\bmtheta,\bmz,\mathring\bmz}^n(S,t') - \ex{\log\rho_{\bmx,\bmtheta,\bmz,\mathring\bmz}^n(S,t')}}^2 } = 0,
    \end{align} 
    where both expectations are taken with respect to the joint distributions of $(\bmx,\bmtheta,\bmz,\mathring\bmz)$.
\end{enumerate}
\end{condition}
We can parametrize the free energy for $\overline\bmy$ by $(\mathring S,t)$ by defining $F_n : (\psd^d)^k \times [0,\infty) \to \bbR$ as 
\begin{align}
    F_n(\mathring S,t) \coloneqq \frac{1}{n}\ex[\big]{\log \rho^n_{\bmx,\bmtheta,\bmz,\mathring\bmz}(\mathring S,t)},
\end{align} 
where again the expectation is taken with respect to the distribution of $(\bmx,\bmtheta,\bmz,\mathring\bmz)$. Notice that the free energy $F_n(t)$ of model~\eqref{eq:y_approx} is then just $F_n(0,t)$, and we identified $F_n(t) \equiv F_n(0,t)$. 

For each $(\mathring S,t) \in (\psd^d)^k \times [0,\infty)$, let us define the potential functions $\Psi : (\psd^d)^k \times (\psd^d)^k \to \bbR$ as
\begin{align} 
    \Psi(Q, S)  &\coloneqq H(S + \mathring S) - \frac{1}{2} \sum_{j=1}^k \inner{S_j}{Q_j} 
    + \frac{t}{2} \sum_{\kappa \in [k]^{q}}  \inner{A_\kappa^\top A_\kappa}{Q_{\kappa_1} \otimes \dots \otimes Q_{\kappa_q}}. \label{eq:potential}
\end{align}

\begin{thm} \label{th:Flimit}
    Assume Condition~\ref{as:Flimit} holds, and define the function $F : (\psd^d)^k \times [0,\infty) \to \bbR$ 
    \begin{align} \label{eq:F_def}
        F(\mathring S, t) \coloneqq \sup_{Q \in (\psd^d)^k}\inf_{S \in (\psd^d)^k} \Psi(Q,S).
    \end{align}
    Then, the free energy $F_n$ of $\overline\bmy$ converges pointwise to $F$  for all $(\mathring S,t) \in (\psd^d)^k \times [0,\infty)$,
    \begin{align}
        \lim_{n \to \infty} F_n(\mathring S,t) = F(\mathring S, t).
    \end{align}
\end{thm}

\subsection{Free energy, overlap and minimum mean-squared error} \label{ss:overlap}
A key feature of the free energy $F_n$ is that the arguments $(\mathring S,t)$ play the role of perturbation parameters with respect to which it is possible to take gradients. These derivatives are known as (expected) \textit{overlaps}, and provide insight into the statistical limits of inference for the model of interest.

In particular, let $\nabla_{\!j}$ denote the partial gradient of $F_n$ with respect to $\mathring S_j$, $j=1,\dots,k$, and $\partial_t$ denote the partial derivative of $F_n$ with respect to $t$. We have the following identities: 
\begin{align}
    \nabla_{\!j} F_n &= \frac{1}{2n} \sum_{i=1}^n  \ex*{ \bbone_{\{\psi(x_i) = j \} }  \ex{ \phi(\theta_i) \mid \bmx,\overline{\bmy}} {\ex{ \phi(\theta_i) \mid \bmx,\overline{\bmy}}}^\top}, \quad j = 1, \dots, k \label{eq:grad_linear}\\
    \partial_t F_n &= \frac{1}{2n} \ex*{ \norm[\big]{\ex{f^{(n)}(\bmx,\bmtheta) \mid \bmx,\overline{\bmy}}}^2 } \label{eq:grad_q_wise}
\end{align}
From identity~\eqref{eq:I_F_identity}, these gradients characterize the Bayes risk under the squared error loss for estimating  $\phi(\theta_1), \dots, \phi(\theta_n)$ and $f^{(n)}(\bmx, \bmtheta)$ from the observed data $(\bmx, \overline{\bmy})$. Note that the identities~\eqref{eq:grad_linear}-\eqref{eq:grad_q_wise} hold uniformly over the relative interior of the domain of $F_n$, $(\pd^d)^k \times (0,\infty)$, and the derivatives can be extended to points on the boundary via continuity, so that they are well-defined on the entirety of $(\psd^d)^k \times [0,\infty)$. 

Meanwhile, the limiting free energy $F$ is convex and thus differentiable almost  everywhere on $(\pd^d)^k \times (0,\infty)$. It can also be verified (see e.g.,~\cite{chen2022}) that the gradients $(\nabla_1 F_n, \dots, \nabla_k F_n, \partial_t F_n)$  are monotone non-decreasing with respect to the partial order on $(\psd^d)^k \times [0,\infty)$ defined by $A \preceq B $ if an only if $B-A \in (\psd^d)^k \times [0,\infty)$. Following arguments in~\cite[Lemma 10]{reeves2019}, it then follows that
\begin{align}
    \nabla_j F_n(\mathring S, t)  \to F(\mathring S, t), \qquad \partial_t F_n(\mathring S, t) \to \partial_t F (\mathring S, t) \label{eq:gradFn_to_gradF}
\end{align}
for all $(\mathring S, t) \in (\pd^d)^k \times (0,\infty)$ at which $F$ is differentiable. In this way, the limiting free energy characterizes the limiting overlaps almost everywhere with respect to the side information $\mathring S$ and signal-to-noise ratio $t$. 

The overlaps for the original model  $(\bmx, \bmy)$ without any side information corresponds to the boundary case $\mathring S = 0$. In some cases, one can leverage structural properties of the original model to verify that the convergence in~\eqref{eq:gradFn_to_gradF} also holds on the boundary. 
However, it may also be the case that limiting overlap is discontinuous on the boundary. In such cases, the limiting Bayes risk can be bounded below in terms of the subdifferential of $F$ evaluated at the boundary point; see~\cite{reeves2020} and~\cite{chen2025a} for more details.

In our applications in the coming sections, we will derive expressions for the free energy assuming there is no auxiliary observation (i.e. $\mathring S=0$), and all results are understood to hold pointwise for all fixed $t \in [0,\infty)$. Our results can be extended directly to the case of generic $\mathring S$ by noticing that our proof techniques hold pointwise for all $\mathring S$. See Appendices~\ref{s:proof_hlr} and~\ref{s:proof_qap} for more details.  

\section{Heteroskedastic spiked tensor model} \label{s:hetero}
The heteroskedastic spiked tensor model is a variation of the low-rank spiked tensor model in which a collection of vectors $\bmtheta \coloneqq (\theta_i)_{i \in [n]}, \theta_i \in \R^d$ gives rise to the Gaussian noise-corrupted observations $\bmy \coloneqq (y_\alpha)_{\alpha \in [n]^q}$,
\begin{align} \label{eq:y_hetero}
    y_{\alpha} = \sqrt{\frac{t}{n^{q-1}}} \inner{b}{\, \theta_{\alpha_1} \otimes \dots \otimes \theta_{\alpha_q}} + \sigma_{\alpha} z_{\alpha}, \quad \alpha \in [n]^q,
\end{align}
where $b \in \bbR^{d^q}$ is a known vector defining a multilinear map. The scalars $\sigma_{\alpha} > 0$ are known parameters that  define the noise variance profile. We also allow for $\sigma_{\alpha}=+\infty$, to represent unobserved entries. To simplify the presentation, we assume that the unknown parameters $\theta_i$ are i.i.d.\ from a compactly supported distribution $\bmeas$ on $\bbR^d$. 

For the matrix case,~\eqref{eq:y_hetero} readily reduces the standard low-rank spiked matrix model by letting $b = \vec \id_d$, so that $(\theta,\theta') \mapsto \inner{b}{\theta \otimes \theta'} = \inner{\theta}{\theta'}$ defines the inner product map. The fundamental limits of the resulting model
\begin{align}
    y_{ij} = \sqrt{\frac{t}{n}} \inner{\theta_i}{\theta_j} + \sigma_{ij} z_{ij}, \quad i,j \in [n]
\end{align}
were previously derived in~\cite{behne2022} assuming that the number of unique entries in $(\sigma_{ij})$ stays bounded as $n \to \infty$, and in~\cite{guionnet:2025low-rank} with the constraint that the matrix $(\sigma_{ij})_{i,j\in [n]}$ be positive definite. 

Our results generalize to the spiked tensor case under arbitrary couplings (i.e. choice of $b$) and a large class of variance profiles.
The key observation is that one can represent \eqref{eq:y_hetero} in terms of the model~\eqref{eq:y_alpha}, by using covariates $\bmx$ to encode the index locations, that is $\bmx \coloneqq (1/n,2/n,\dots,1)$ for all $n \in \bbN$. Our approach allows to characterize the mutual information for this model provided that the covariance profile satisfies the following regularity condition.

For $n \in \bbN$ 
define the piecewise-constant functions as $\Omega_n \colon [0,1]^q \to [0, \infty)$ with $\Omega_n(0) = 0$ and  
\begin{align}
\Omega_n(x_1,\dots,x_q)\coloneqq \sigma^{-1}_{\lceil n x_1 \rceil \cdots \lceil n x_q \rceil}
\end{align}
for all $x \in (0,1]^q$ where  $\lceil x \rceil \coloneqq \inf \set{y \in \bbN : x \leq y}$. Here we use the convention that $\sigma_\alpha = +\infty$ is mapped to zero, and refer to functions $\Omega_n$ as the inverse noise variance function. 

\begin{condition}\label{as:hlr} 
The inverse noise variance $\Omega_n$ is bounded uniformly for all $n$.  Furthermore, there exists a bounded measurable function $\Omega : [0,1]^q \to [0,\infty)$ such 
\begin{align}
    \lim_{n \to \infty} \int_{[0,1]^q} \abs[\big]{\Omega(x) - \Omega_n(x)} \dd x = 0.
\end{align}
\end{condition}
In words, this condition ensures that the variance profiles converge to a limit that is continuous, except on a set of arbitrarily small measure. 

To make the reduction to the $q$-way interaction model explicit, we note that the observations~\eqref{eq:y_hetero} can be put into one-to-one correspondence, and are therefore statistically equivalent, to observations
\begin{align}
    y_\alpha = \sqrt{\frac{t}{n^{q-1}}} \; \Omega_n(x_{\alpha_1},\dots,x_{\alpha_q}) \inner[\big]{b}{\, \theta_{\alpha_1} \otimes \dots \otimes \theta_{\alpha_q}} + z_\alpha, \quad \alpha \in [n]^q.
\end{align}

Our main result for this section states that, under Condition~\ref{as:hlr}, the asymptotic free energy of the heteroskedastic spiked tensor model can be expressed as the solution to a variational problem whose structure is analogous to~\eqref{eq:F_def}. A proof is presented in Appendix~\ref{s:proof_hlr}.

For a given triple $(P, \Omega, b)$, we define the maps  $D,J : \cS([0,1]) \to L^2([0,1])$, 
\begin{align}
    [H(\sfS)](x) &= \div[\big]{[\sfS(x)]_\#\bmeas * \gaus}{\gaus}, \\
    [J(\sfS)](x) &= \int_{[0,1]^{q-1}} \big[\Omega(x,x_2,\dots,x_q)\big]^2 \inner{bb^\top}{\sfS(x_2) \otimes \dots \otimes \sfS(x_q)} \dd (x_2,\dots,x_q),
\end{align}
where  $\cS([0,1])$ is the space of bounded measurable functions $\sfS : [0,1] \to \psd^d$, and the potential function  $\Psi : \cS([0,1]) \times \cS([0,1]) \to \bbR$,
\begin{align}
    \Psi(\sfQ,\sfS) \coloneqq \int_{[0,1]} [H(\sfS)](x) - \frac{1}{2} \inner{\sfS(x)}{\sfQ(x)} + \frac{t}{2} [J(\sfQ)](x) \dd x.
\end{align}
\begin{thm} \label{th:hlr} 
Assume Condition~\ref{as:hlr} holds. Then, for all $t \in [0,\infty)$, the free energy $F_n(t)$ for the heteroskedastic spiked tensor model~\eqref{eq:y_hetero} satisfies
    \begin{align}
        \lim_{n \to \infty} F_n(t) = \adjustlimits \sup_{\sfQ \in \cS([0,1])} \inf_{\sfS \in \cS([0,1])} \Psi(\sfQ,\sfS).
    \end{align}
\end{thm}
\begin{remark}
Theorem~\ref{th:hlr} applied to the $q=2$ case allows to recover the results in \cite[Theorem 1]{behne2022} and \cite[Theorem 2.10]{guionnet:2025low-rank} as special cases. For the former, it is sufficient to notice that when the number of unique variance types is finite then Condition~\ref{as:hlr} is immediately satisfied. For the latter, the function $\Omega$ can be understood as the positive semidefinite kernel corresponding to the uniform limit of the positive definite matrices $(\sigma_{ij}^{-1})_{i,j \in [n]}$.
\end{remark}

\section{\texorpdfstring{$q$-adic}{q-adic} assignment problem} \label{s:qap}
The $q$-adic assignment problem is a classical matching problem that was proposed by Lawler~\cite{lawler:1963the-quadratic} as a higher-order generalization of the quadratic assignment problem of Koopmans and Beckmann~\cite{koopmans1957}. Given two real-valued arrays $(\bmv,\bmy) \coloneqq (v_\alpha,y_\alpha)_{\alpha \in [n]^q}$, we are tasked to find a permutation $\sigma \in \sym([n])$ that maximizes the alignment between $\bmv$ and $\bmy$, as measured by their inner product:
\begin{align}
    \underset{\sigma \in \sym([n])}{\text{maximize}} \quad \sum_{\alpha \in [n]^q} \inner{v_{\sigma(\alpha)}}{y_\alpha},
\end{align}
where the action of $\sigma$ on index elements $\alpha \in [n]^q$ is given by  $\sigma(\alpha) \coloneqq (\sigma(\alpha_1),\dots,\sigma(\alpha_q))$. 

Here, we consider a probabilistic formulation of the $q$-adic assignment problem, generalizing the problem statement of~\cite{yang2024} beyond the quadratic case. We assume $\bmy$ is a noisy observation of $\bmv$, whose indices have been permuted by some unknown random permutation $\sigma \in \sym([n])$, that is 
\begin{align}
    y_\alpha = v_{\sigma(\alpha)} + z_\alpha, \quad \alpha \in [n]^q,
\end{align}
for $z_\alpha \sim_\iid \normal(0,1) \Perp \sigma \sim \unif(\sym([n]))$. The task is to recover $\sigma$ conditionally on knowledge of $(\bmv,\bmy)$. Additionally, we assume that the array $\bmv$ has the following separable structure: there exist points $\bmu \coloneqq (u_i)_{i \in [n]}$, $u_i \in [0,1]$, and a continuous function $f : [0,1]^q \to \bbR$ such that
\begin{align}
    v_{\alpha} = \sqrt{\frac{t}{n^{q-1}}}f(u_{\alpha_1},\dots,u_{\alpha_q}), \quad \alpha \in [n]^q.
\end{align}
We identify the signal $\bmtheta \coloneqq (\theta_1,\dots,\theta_n)$ with the randomly permuted atoms $\bmu$, that is $\bmtheta \coloneqq (u_{\sigma(1)},\dots,u_{\sigma(n)})$. We obtain the model
\begin{align} \label{eq:qap}
    y_\alpha = \sqrt{\frac{t}{n^{q-1}}} f(\theta_{\alpha_1},\dots,\theta_{\alpha_q}) + z_\alpha, \quad \alpha \in [n]^q.
\end{align}
The mutual information $I(\sigma; \bmv, \bmy) = H(\sigma) - H(\sigma \mid \bmv, \bmy)$ provides a natural measure how much can be learned about the unknown permutation $\sigma$ from the observed pair. Using our results, we  characterize the leading order terms in the mutual information by providing exact formula for  $\lim_{n \to \infty} \frac{1}{n} I(\sigma; \bmv, \bmy)$.

The basic idea is to interpret the $q$-adic assignment problem~\eqref{eq:qap} above as an instance of the general $q$-wise interaction model~\eqref{eq:y_alpha}, where the unknown signal $\bmtheta$ is the randomly permuted vector $\bmu$.  In the case that $\bmu$ is known and has distinct entries, there is a bijection between $\sigma$ and $\bmtheta$ and so the problem of recovering $\sigma$ is equivalent to the problem of recovering $\bmtheta$. Interestingly, it turns out that this equivalence holds generally without requiring any identifiably  conditions on $\bmu$. This is formalized in the statement below.

\begin{prop} \label{prop:qap_info_equivalence}
    Consider the $q$-adic assignment problem~\eqref{eq:qap}. Then, $I(\sigma; \bmv,\bmy) = I(\bmtheta; \bmy \mid \bmu)$.
\end{prop}
With this, we can move to characterizing the asymptotic free energy for $q$-adic assignment problem. Our result applies assuming the following condition on $f:[0,1]^q \to \bbR$ and $\bmu$ holds.

\begin{condition}\label{as:qap}
The empirical measure $\frac{1}{n}\sum_{i=1}^n \delta_{u(i)}$ converges weakly to a distribution $\bmeas$ supported on the unit interval. Furthermore, there exist continuous functions $(\phi_i)_{i \in \bbN}$, $\phi_i : [0,1] \to \bbR$, and coefficients $(b_\alpha)_{\alpha \in \bbN^q}$ such that, for every $\eps > 0$, there exists an integer $d$ such that
    \begin{align}
\sup_{(\theta_1,\dots,\theta_q)\in [0,1]^q} \abs[\Big]{ f(\theta_1,\dots,\theta_q) - \sum_{\alpha \in [d]^q} b_\alpha  \left(\phi_{\theta_1}(\alpha_1) \cdots \phi_{\alpha_q}(\theta_q) \right) }\leq\eps.  \label{eq:f_approx} 
\end{align}
\end{condition}

The key idea for characterizing the statistical limits of~\eqref{eq:qap} is that under Condition~\ref{as:qap} we can construct and $\eps$-dependent sequence of $(d,k)$-multilinear approximations $\hat f^{(n)} \in \cF_{d,k}$ to $\bmw$ with uniform control over the approximation error. For this model covariates $\bmx$ will not be necessary, and we can ignore the role of $k$. All details and a proof of Theorem~\ref{th:qap} are contained in Appendix~\ref{s:proof_qap}.

For each $\eps >0$, we set $d$ to be the smallest positive integer such that \eqref{eq:f_approx} holds with error $\eps$. We set $B^{(\eps)} \coloneqq (b_\alpha)_{\alpha \in [d]^q}$ and define the vector-valued map $\phi^{(\eps)} : [0,1] \to \bbR^d$ as $\phi^{(\eps)} : \theta \mapsto \big(\phi_1(\theta), \dots, \phi_d(\theta)\big)$. For each $\eps>0$, we construct the $d$-linear approximation model
\begin{align} \label{eq:y_eps}
    y_{\alpha}^{(\eps)} = \sqrt{\frac{t}{n^{q-1}}} \inner[\big]{ \vec B^{(\eps)}}{\phi(\theta_{\alpha_1}) \otimes \dots \otimes \phi(\theta_{\alpha_q})} + z_{\alpha}, \quad \alpha \in [n]^q.
\end{align}
For each $\eps>0$ and $S \in \psd^d$, we define the linear channel relative entropy $H^{(\eps)} :\psd^d \to \bbR$ as $H^{(\eps)}(S) = \div[\big]{ (S^{1/2}\phi^{(\eps)})_\# \bmeas * \gaus }{\gaus}$, as well as the potential function $\Psi^{(\eps)} : \psd^d \times \psd^d \to \bbR$,
\begin{align}
    \Psi^{(\eps)}(Q,S) \coloneqq H^{(\eps)}(S)  - \frac{1}{2} \inner{Q}{S} + \frac{t}{2}\inner[\big]{\vec(B^{(\eps)})\vec(B^{(\eps)})^\top}{Q^{\otimes q}}.
\end{align}
For each $\eps > 0$ and let $F^{(\eps)} : [0,\infty) \to \bbR$ be defined pointwise as the solution to the variational problem
\begin{align}
    F^{(\eps)}(t) \coloneqq \adjustlimits \sup_{Q \in \psd^d}\inf_{S \in \psd^d} \Psi^{(\eps)}(Q,S).
\end{align}
\begin{thm}\label{th:qap}
Assume that Condition~\ref{as:qap} holds. Further, assume that for all $\eps > 0 $ model \eqref{eq:y_eps} satisfies Condition~\ref{as:Flimit}. Then, for all $t \in [0,\infty)$, the free energy $F_n$ for \eqref{eq:qap} satisfies 
    \begin{align}
        \lim_{n \to \infty}F_n(t) = 
         \lim_{\eps \to 0} F^{(\eps)}(t),
    \end{align}
 where both limits exist and are finite. 
\end{thm}
\begin{remark} 
Theorem~\ref{th:qap} generalizes \cite[Theorem 4.4]{yang2024}, which gives a formula for the asymptotic free energy under the additional assumption that $q=2$ and $f$ is a positive semidefinite kernel. Note that if $f$ is positive semidefinite, then the functional approximation in \eqref{eq:f_approx} follows directly from Mercer's theorem. 
Remaining in the $q=2$ setting, Condition~\ref{as:qap} corresponds to a generalization of Mercer's theorem for indefinite and asymmetric kernels, which exists under mild smoothness assumptions; see~\cite{jeong:2024extending} for more details.
\end{remark}

% BIBLIOGRAPHY
\bibliographystyle{plain}
\bibliography{library,galenbib}

\clearpage

% APPENDIX
\appendix
\section{Proof of main results} \label{s:proof_main}
\subsection{Proof of Theorem~\ref{th:approx}} \label{ss:proof_approx}
Before presenting the proof, we recall the HWI inequality~\cite{otto:2000generalization}, that will play a key role in the arguments. We will state it following  the presentation of~\cite{gentil:2020an-entropic}.
Let $\mathfrak{m}$ be a probability measure $\mathfrak{m}(\dd x) = \dd x \, e^{-V(x)}$, for $\dd x$ denoting the standard Lebesgue measure on $\bbR^n$ and $V \in C^{\infty}(\bbR^n)$ such that $\operatorname{Hess}(V) \geq \kappa \id_n$. Let $\mu$ be a probability measure that is absolutely continuous with respect to $\mathfrak{m}$, and denote the associated Radon-Nikodym derivative by $\rho$. We define the relative Fisher information divergence $\cI(\mu \; \Vert \; \mathfrak{m})$ as
\begin{align} \label{eq:as_fisher}
    \cI(\mu \: \Vert \: \mathfrak{m}) \coloneqq \ex{\norm{\nabla \log \rho}^2} = \int \frac{\norm{\nabla \rho}^2}{\rho} \dd \mathfrak{m}.
\end{align}
The HWI inequality states that, for two probability measures $\mu,\nu$ that are absolutely continuous with respect to reference measure $\mathfrak{m}$, one has
\begin{align} \label{eq:hwi}
    \div{\mu}{\mathfrak{m}} - \div{\nu}{\mathfrak{m}} \leq W_2(\mu, \nu) \sqrt{\cI(\mu \: \Vert \: \mathfrak{m})} - \frac{\kappa}{2} W_2^2(\mu, \nu).
\end{align}
With this technical tool in hand, we are ready to present our proof.

\paragraph{Proof of Theorem~\ref{th:approx}.}
For any $n \in \bbN$, we will consider as reference the Gaussian measure $\gaus$ defined on the observation space $(\bbR^m)^{N}$. With this choice, it is easily seen that $\kappa = 1$.

We will consider a function $f : \cX^q \times \Theta^q \to \bbR^m$ and assume it satisfies Condition~\ref{cond:eps_approx}. Let $\hat f \in \cF$ be the associated $(d,k)$-multilinear approximation, for some $d,k \in \bbN$. For any $\delta > 0$, then there exists some $\bar n \in \bbN$ such that for all $n \geq \bar n$
\begin{align}
    \ex*{W_2^2\Big( \big[\sqrt{t}f^{(n)}_\# \bmeas^n * \gaus\big](\blank \mid \bmx), \,  \big[\sqrt{t} \hat f^{(n)}_\# \bmeas^n * \gaus\big](\blank \mid \bmx)} \leq n\delta .
\end{align}

Notice that, seen as a function of $\bmy \coloneqq \sqrt{t} f^{(n)}(\bmx,\bmtheta) + \bmz$, the integral of the Hamiltonian
\begin{align}
    \rho^n(\bmy \mid \bmx) \coloneqq \int \exp\{\cH_{\bmx,\bmtheta,\bmz}^n(\bmtheta' ; t)\} \dd \bmeas^n(\bmtheta ' \mid \bmx)
\end{align}
is the density with respect to the Gaussian measure of $\bmtheta \mid (\bmx,\bmy)$, so that
\begin{align}
    \bmeas^n(\dd \bmtheta \mid \bmx,\bmy) = \bmeas^n(\dd \bmtheta \mid \bmx) e^{-\frac{t}{2}\norm{f^{(n)}(\bmx,\bmtheta)}^2} \exp\{\inner{\sqrt{t}f^{(n)}(\bmx,\bmtheta)}{\bmy} - \log \rho^n(\bmy \mid \bmx)\}.
\end{align}
From an application of Tweedie's formula~\cite{robbins:1992an-empirical},  then, the gradient of $\log \rho^n$ yields the conditional expectation 
\begin{align}
    \nabla_{\!\bmy} \log \rho^n(\bmy \mid \bmx) = \ex{\sqrt{t}f^{(n)}(\bmx,\bmtheta) \mid \bmx,\bmy}.
\end{align}
It follows that the expected Fisher information divergence $\ex[\big]{\cI\big( [\sqrt{t}f^{(n)}_\#\bmeas^n](\blank \mid \bmx) \:\Vert\: \gaus \big)}$ can be bounded in terms of the second moment of $f^{(n)}$, since from the above display
\begin{align}
    \ex[\big]{\cI\big( [\sqrt{t}f^{(n)}_\#\bmeas^n](\blank \mid \bmx) \:\Vert\: \gaus \big)} 
    &= t\ex[\big]{\norm{\ex{f^{(n)}(\bmx,\bmtheta) \mid \bmx,\bmy}}^2} \\
    &\leq t \ex{f^{(n)}(\bmx,\bmtheta)},
\end{align}
where the last line is bounded above by $tnC$ for some constant $C>0$ and $n$ sufficiently large. By Condition~\ref{cond:eps_approx}, it follows that for $n$ sufficiently large also the $(d,k)$-multilinear approximation $\hat f$ satisfies
\begin{align}
    \ex[\big]{\cI\big( [\sqrt{t}\hat f^{(n)}_\#\bmeas^n](\blank \mid \bmx) \:\Vert\: \gaus \big)} \leq tC(n + \delta).
\end{align}
Then, an application of the HWI inequality~\eqref{eq:hwi} on both sides yields that for all $n$ sufficiently large
\begin{align}
    \frac{1}{n}\abs{F_n - \hat F_n} \leq t\left( \sqrt{C\delta + \delta^2} + \frac{\delta^2}{2} \right).
\end{align}
Since $\delta>0$ was chosen arbitrarily, we conclude that
\begin{align}
    \limsup_{n \to \infty} \big | F_n - \hat{F}_{n} \big| \le \eps
\end{align}
for any $\eps > 0$, as desired.

\subsection{Proof of Theorem~\ref{th:Flimit}} \label{ss:proof_Flimit}
We divide our proof in three parts. First, we present the MTP model and the main result for its asymptotic free energy of Chen et al.~\cite{chen2022}, which will play a central role in our proof. Then, we employ a lifting argument to show that any $(d,k)$-multilinear model can be seen as a special MTP in a lifted space. Finally, we establish some orthogonal invariance properties of the free energy (as well as the potential function that characterizes it) and leverage them to establish Theorem~\ref{th:Flimit}.
\paragraph{Preliminaries: Free energy of MTP.}
We begin by presenting the order-$q$ MTP model, defined as the collection of observations $\bmy \coloneqq (y_{\alpha})_{\alpha \in [n]^q}$,
\begin{align} \label{eq:y_mtp_original}
    y_\alpha = \sqrt{\frac{t}{n^{q-1}}} B \big(\eta_{\alpha_1} \otimes \dots \otimes \eta_{\alpha_q}\big) + z_\alpha, \quad \alpha \in [n]^q,
\end{align}
where $z_\alpha \sim_\iid \normal(0,\id_m)$ independently of $\bmeta \coloneqq (\eta_1,\dots,\eta_n) \in (\bbR^d)^n$ and $B \in \bbR^{m \times d^q}$ is a known matrix. 

Similarly to the statement of Theorem~\ref{th:Flimit}, we augment the MTP observations with an auxiliary linear channel observation $\mathring\bmy$ parametrized by $\mathring S \in \psd^d$,
\begin{align} \label{eq:y_linear_appendix}
    \mathring y_i = \mathring S^{1/2} \eta_i + \tilde z_i, \quad i \in [n],
\end{align}
where $\tilde z_n \sim_\iid \normal(0,\id_d)$ independently of everything else. As in Section~\ref{ss:F_limit}, we let $H : \psd^d \to \bbR$ be the pointwise limit of the relative entropy functions $H_n$,
\begin{align}
    H_n : \mathring S \mapsto \frac{1}{n}\ex[\Big]{\log \int \exp\{ \mathring\cH^n_{\bmx,\bmtheta,\mathring\bmz}(\bmtheta'; \mathring S)\}  \, \bmeas^n(\dd\bmtheta' \mid \bmx) },
\end{align}
where we have assumed that $\eta_i = \phi(\theta_i)$ for all $i \in [n]$ and some bounded map $\phi : \Theta \to \bbR^d$. The Hamiltonian $\mathring \cH^n$ above is defined as in~\eqref{eq:hamiltonian_linear} by setting $k=1$. Before getting into our proof, we state in our notation the free energy convergence result of~\cite{chen2022} for the MTP model, which we will play a key role in our proof. 
\begin{thm}[\cite{chen2022}, Theorem 1.1] \label{th:chen}
    Consider observations $(\bmy,\mathring \bmy)$ defined in~\eqref{eq:y_mtp_original} and~\eqref{eq:y_linear_appendix}, and assume Condition~\ref{as:Flimit} holds. Then, for all $(\mathring S,t) \in \psd^d \times [0,\infty)$, the associated free energy $F_N(\mathring S,t)$ satisfies
    \begin{align}
        F_n(\mathring S,t) \to_{n \to \infty} \adjustlimits \sup_{Q \in \psd^d}\inf_{S \in \psd^d} \Big\{ H(S+\mathring S) - \frac{1}{2} \inner{Q}{S} + \frac{t}{2} \inner{B^\top \!B}{Q^{\otimes q}} \Big\}.
    \end{align}
\end{thm}
\begin{remark}
    Compared to the definition of $\Psi$ appearing in~\cite{chen2022}, some terms appearing in the potential function are scaled by a factor of two. This is due to the different normalization conventions employed in the definition of the MTP $\eqref{eq:y_mtp_original}$, and the two formulations are completely equivalent.
\end{remark}

\paragraph{Step 1: Lifting argument.}
The MTP~\eqref{eq:y_mtp_original} is readily seen to be a special case of a $(d,k)$-multilinear model where $k=1$. In our proof, we show that one can go the other way and leverage prior convergence results to characterize the free energy of~\eqref{eq:y_approx}. For the reader's convenience, we recall the explicit definition of a $(d,k)$-multilinear model, consisting of observations
\begin{align} \label{eq:y_dk_appendix}
    y_\alpha = \sqrt{\frac{t}{n^{q-1}}} A_{\psi(x_{\alpha_1})\cdots \psi(x_{\alpha_q})} \big( \phi(\theta_{\alpha_1}) \otimes \dots \otimes \phi(\theta_{\alpha_q}) \big) + z_\alpha, \quad \alpha \in [n]^q,
\end{align}
for two bounded functions $\psi : \cX \to [k]$, $\phi : \Theta \to \bbR^d$ and known matrices $A_\kappa \in \bbR^{m \times d^q}$, $\kappa \in [k]^q$.

We let $\bmw \coloneqq (\psi(x_1),\dots,\psi(x_n)) \in [k]^n$, and for $j \in [k]$ we use $e_j$ to denote the $j$-th standard basis vector of $\bbR^k$. We use $\bmomega \in \bbR^{k \times n}$ to denote a matrix whose rows are one-hot encoding of $\bmw$, i.e. $\omega_i = e_j \iff w_i=j$. We define the lifted features
\begin{align}
\widehat{\eta}_i & =  \omega_i \otimes \eta_i \in \bbR^{dk}. \label{eq:eta_bar}
\end{align}
From the properties of the Kronecker product, for every $q \in \N$ there exist a $(dk)^q \times (dk)^q$ permutation matrix $\Pi_q$ such that 
\begin{align} \label{eq:kron_commute}
  \Pi_q \, (\widehat\eta_{\alpha_1} \otimes \dots \otimes \widehat\eta_{\alpha_q})
   & = (\omega_{\alpha_1} \otimes \cdots \otimes \omega_{\alpha_q}) \otimes (\eta_{\alpha_1} \otimes \dots \otimes \eta_{\alpha_q})
\end{align}
for all $\alpha \in [n]^q$. Notice that $\Pi_1$ is trivially the identity matrix.

We introduce a matrix $\wA \in \bbR^{m \times (dk)^q}$ by horizontally stacking all matrices $A_\kappa$, $\kappa \in [k]^q$, in ascending order, and permuting on the right according to $\Pi_q$. That is
\begin{align}
\wA &\coloneqq \begin{bmatrix}
    A_{1\cdots1} & A_{1\cdots2} & \cdots & A_{1\cdots k} & A_{2\cdots1} & \cdots &  A_{k \cdots k}
\end{bmatrix}\, \Pi_q  \in \bbR^{m \times (dk)^q}.
\end{align}
With this, we see that \eqref{eq:y_dk_appendix} can be embedded into a lifted MTP model with parameters in $\bbR^{dk}$,
\begin{align} \label{eq:y_lift}
y_{\alpha} = \sqrt{\frac{t}{n^{q-1}}} \wA \,\big(\widehat\eta_{\alpha_1} \otimes \dots \otimes \widehat\eta_{\alpha_q}\big) + z_\alpha, \quad \alpha \in [n]^q.
\end{align}
We introduce the linear channel relative entropy $\widehat{H}_n : \psd^{dk} \to \bbR$ defined as
\begin{align}
    \widehat{H}_n(S) = \frac{1}{n} \ex[\Big]{ \log \int \exp \Big\{ \sum_{i=1}^n \inner[\big]{\omega_i \otimes \phi(\theta_i')}{ \,S^{1/2}\widehat{\eta}_i + \widehat z_i } + \norm{\omega_i \otimes \phi(\theta_i')}^2 \Big\} \, \bmeas^n(\dd\bmtheta'\mid \bmx) },
\end{align}
where $\widehat z_i \sim_\iid \normal(0,\id_{dk})$ independent of everything else. Note that as long as Condition~\ref{as:Flimit} is satisfied then $\widehat{H}_n$ converges pointwise to a well-defined and continuously differentiable limit, which we denote as $\widehat H : \psd^{dk} \to \bbR$. By Theorem~\ref{th:chen}, the free energy $\widehat F_n$ for~\eqref{eq:y_lift} converges pointwise to the solution to the variational problem
\begin{align}
    \widehat F_n(R,t) = \adjustlimits \sup_{Q \in \psd^{dk}}\inf_{S \in \psd^{dk}} \widehat \Psi(Q,S),
\end{align}
where $\widehat \Psi(Q,S) : \psd^{dk}\times\psd^{dk}\to\bbR$ is defined for all $\mathring S \times t \in \psd^{dk} \times [0,\infty)$ as
\begin{align}
    \widehat \Psi(Q,S) = \widehat H(S + \mathring S) - \frac{1}{2}\inner{Q}{S} - \frac{t}{2} \inner{\widehat{A}^\top \!\widehat A}{\,Q^{\otimes q}},
\end{align}
completing the reduction a generic $(d,k)$-multilinear approximation to an instance of the MTP model.

\paragraph{Step 2: Orthogonal invariance of free energy.}
Before proceeding, we state the following intermediate results, which are proved in Appendix~\ref{ap:approx_lemmas}.
\begin{lem} \label{lem:F_invariance}
    Let $\kappa \in [k]^{q}$. Define matrices 
    \begin{align}
        U_\kappa \coloneqq \Pi_q^\top \big( e_{\kappa_1}e_{\kappa_1}^\top \otimes \cdots \otimes e_{\kappa_q}e_{\kappa_q}^\top \otimes \id_{d^q} \big) \Pi_q
    \end{align}
    as well as mapping $\cP_q : \bbR^{(dk)^q \times (dk)^q} \to \bbR^{(dk)^q \times (dk)^q}$
    \begin{align} \label{eq:projectors}
        \cP_q : M \mapsto \sum_{\kappa \in [k]^q} U_\kappa M U_\kappa.
    \end{align}
    Then the free energy $F_n^\cP$ associated with observation model $\bmy^\cP \coloneqq (y_\alpha^\cP)_{\alpha \in [n]^q}$
    \begin{align}
        y_\alpha^\cP = \sqrt{\frac{t}{n^{q-1}}} \cP_q(\widehat{A}^\top\widehat{A})^{1/2} \: (\widehat\eta_{\alpha_1} \otimes \dots \otimes \widehat \eta_{\alpha_q}) + z_\alpha^\cP,
    \end{align}
    for $z_\alpha^\cP \sim_\iid \normal(0, \id_{dk})$ independent of everything else, satisfies the equality
    \begin{align}
        \widehat F_n(\mathring S,t) = F_n^\cP(\mathring S,t)
    \end{align}
    for all $(\mathring S,t) \in \psd^{dk} \times [0,\infty)$.
\end{lem}

\begin{lem} \label{lem:potential_diagonal}
Consider potential function $\Psi^\cP:\psd^{dk}\times\psd^{dk}\to\bbR$ defined as
\begin{align}\label{eq:potential_lift}
    \Psi^\cP(Q,S) = \widehat H(S + \mathring S) - \frac{1}{2}\inner{Q}{S} - \frac{t}{2} \inner[\big]{\cP_q(\widehat{A}^\top \!\widehat A)}{\,Q^{\otimes q}}.
\end{align}
Fix any two matrices $\widehat{Q},\widehat{S} \in \psd^{dk}$, and let $Q,S \in (\psd^d)^k$ be two collection of matrices such that $Q_j = (e_j^\top \otimes \id_d) \widehat{Q}(e_j \otimes \id_d)$ and $S_j = (e_j^\top \otimes \id_d) \widehat{S}(e_j \otimes \id_d)$, for all $j=1,\dots,k$. Let $\cP_1$ be the orthogonal projector on $\bbR^{dk \times dk}$ defined according to~\eqref{eq:projectors} for $q=1$. Then,
\begin{align}
    \Psi^\cP (\cP_1(\widehat{Q}),\cP_1(\widehat{S})) = \Psi(Q,S),
\end{align}
where $\Psi : (\psd^d)^k \times (\psd^d)^k \to \bbR$ the potential defined in~\eqref{eq:potential}. 
\end{lem}

We are now ready to finish the proof of our main result.
From Lemma~\ref{lem:F_invariance}, it follows that the large-$n$ limit of $\widehat F_n$ is exactly equal to the limit of $F_n^\cP$ for all $(\mathring S,t)$, so that by Theorem~\ref{th:chen}
\begin{align}
    \lim_{n \to \infty} \widehat F_n(\mathring S,t) = \adjustlimits \sup_{\widehat{Q} \in \psd^{dk}} \inf_{\widehat{S} \in \psd^{dk}} \Psi^\cP (\widehat{Q},\widehat{S}).
\end{align}
All that is left to complete the proof, then, is to show the equality
\begin{align}
    \adjustlimits \sup_{Q \in (\psd^{d})^k} \inf_{S \in (\psd^{d})^k}  \Psi(Q,S) = \adjustlimits \sup_{\widehat{Q} \in \psd^{dk}} \inf_{\widehat{S} \in \psd^{dk}} \Psi^\cP(\widehat{Q},\widehat{S}).
\end{align}
Finally, from Lemma~\ref{lem:potential_diagonal}, the proof is completed by showing that 
\begin{align}
    \adjustlimits \sup_{\widehat{Q} \in \psd^{dk}} \inf_{\widehat{S} \in \psd^{dk}} \Psi^\cP(\cP_1(\widehat{Q}),\cP_1(\widehat{S})) = \adjustlimits \sup_{\widehat{Q} \in \psd^{dk}} \inf_{\widehat{S} \in \psd^{dk}} \Psi^\cP(\widehat{Q},\widehat{S}).
\end{align}
Using again Lemma~\ref{lem:F_invariance} in the case $q=1$ (or equivalently Lemma~\ref{lem:entropy_proj} in Appendix~\ref{ap:approx_lemmas}), we have that $\widehat H(\widehat{S})= \widehat H(\cP_1(\widehat{S}))$, and we can write
\begin{align}
    \widehat H(\widehat{S} + \mathring S) - \frac{1}{2} \inner{\widehat{Q}}{\widehat{S}} = \widehat H(\cP_1(\widehat{S} + \mathring S)) - \frac{1}{2} \inner{\cP_1(\widehat{Q})}{\cP_1(\widehat{S})} - \frac{1}{2} \inner{\cP_1^\perp(\widehat{Q})}{\cP_1^\perp(\widehat{S})}
\end{align}
where $\cP_1^\perp \coloneqq \operatorname{Id} - \cP_1$ is the orthogonal complement of $\cP_1$. From this, it follows that for all $\widehat{Q}$ such that $\cP_1^\perp(\widehat{Q})\neq 0$, there exists a sequence of matrices $\widehat S_n$, $\cP_1^\perp(\widehat{S_n})\neq 0$, such that
\begin{align}
    \sup_{n \in \bbN} \widehat H(\cP_1(\widehat S_n + \mathring S)) < \infty \quad \text{and}\quad\lim_{n \to \infty} \inner{\cP_1^\perp(\widehat{Q})}{\cP_1^\perp(\widehat{S_n})} = + \infty,
\end{align}
from which we deduce that, for any such $\widehat Q$,
\begin{align}
    \inf_{\widehat{S} \in \psd^{dk}} \widehat H(\widehat{S} + \mathring S) - \frac{1}{2} \inner{\widehat{Q}}{\widehat{S}} = -\infty.
\end{align}
This implies that it must be the case that
\begin{align}
\adjustlimits \sup_{\widehat{Q} \in \psd^{dk}} \inf_{\widehat{S} \in \psd^{dk}} {\Psi}^\cP(\widehat{Q},\widehat{S}) = \adjustlimits \sup_{\widehat{Q} \in \psd^{dk}} \inf_{\widehat{S} \in \psd^{dk}} {\Psi}^\cP(\cP_1(\widehat{Q}),\cP_1(\widehat{S})).
\end{align}
Putting everything together, we have that 
\begin{align}
    \lim_{n \to \infty} F_n(\mathring S,t) = \adjustlimits \sup_{Q \in (\psd^{d})^k} \inf_{S \in (\psd^{d})^k} \Psi(Q,S),
\end{align}
thus concluding the proof of Theorem~\ref{th:Flimit}.

\section{Proofs for heteroskedastic spiked tensor model}\label{s:proof_hlr}
Before getting into the proof, we introduce some necessary concepts and notation. For the set $[0,1]^q$ endowed with Lebesgue measure, let $\vol(B)$ be the Lebesgue measure associated with any measurable subset $B \subseteq [0,1]^q$. Let $f \in L^2([0,1]^q)$ and $\cI$ be a measurable partition of $[0,1]$ into finitely many non-null sets. Denote its cardinality by $|\cI|$. 

We say that a sequence of measurable partitions $(\cI_k)_{k \in \bbN}$ of $[0,1]$ \textit{eventually separates points} if, for every distinct points $x\neq y \in [0,1]$, there are at most finitely many $n \in \N$ for which $x,y \in I_{k,j}$ for a single partition set $I_{k,j} \in \cI_k$, $j=1,\dots,|\cI|$. 

Define the \textit{stepping operator} under $\cI$ as the map $\cP_\cI : L^2([0,1]^q) \to L^2([0,1]^q)$,
\begin{align}
    \cP_\cI f(x_1,\dots,x_p) = \frac{1}{\vol(R)} \int_R f(x_1,\dots,x_q) \dd(x_1,\dots,x_q)
\end{align}
whenever $(x_1,\dots,x_q) \in R$ for some $R \in \cI^q$, i.e. $R$ is a $q$-wise cartesian product of elements of $\cI$. It is readily verified that for any partition $\cI$ the operator $\cP_\cI$ is an orthogonal projection of $L^2([0,1]^q)$ onto the subspace of square-integrable step functions that are constant over elements of $\cI^q$.

A key observation is that, since any bounded measurable function $f$ can be approximated almost everywhere and hence in the $L^2$ sense by uniform integrability, any sequence of partitions $(\cI_k)_{k \in \bbN}$ that eventually separates points induces a sigma-algebra that is fine enough to approximate $f$ via a sequence of steppings.

\begin{prop}[\cite{lovasz:2012large}, Proposition 9.8] \label{prop:stepping_convergence}
    Let $(\cI_k)_{k \in \bbN}$ be a sequence of measurable partitions of $[0,1]$ that eventually separates points, and $f \in L^2([0,1]^q)$ be a bounded function. Then the sequence of steppings $\cP_{\cI_k} f$ converges almost everywhere to $f$ as $n \to \infty$.
\end{prop}

\subsection{Proof of Theorem~\ref{th:hlr}} The main idea of our approach is to exploit a sequence of steppings of the limiting inverse-variance function $\Omega$ to decouple the variance structure from problem dimension $n$. The proof proceeds in two main steps: first, we use the convergence of steppings to construct a sequence of approximation models to~\eqref{eq:y_hetero} whose free energy we can control; secondly, we show that the free energy of the sequence of approximating models is well-described by the functional variational formula prescribed by Theorem~\ref{th:hlr}. Note that we will prove the result for the no-side-information case (i.e. $\mathring S=0)$, but all our arguments are translation invariant so that our result holds over the whole range $\cS([0,1]) \times [0,\infty)$.

\paragraph{Step 1: Approximating sequence.} For every $k \in \bbN$, define a partition $\cI_k$ of the unit interval $[0,1]$, constructed so that the sequence $(\cI_k)_{k \in \bbN}$ eventually separates points. Let us assume without loss of generality that $|\cI_k|=k$ and $\cI_{k+1}$ is a refinement of $\cI_k$, for all $k \in \bbN$.

Condition~\ref{as:hlr} together with Proposition~\ref{prop:stepping_convergence} implies that, for every $\eps > 0$, there exists some integer $k \in \bbN$ such that the approximating observations
\begin{align}
    y^{(k)}_\alpha = \sqrt{\frac{t}{n^{q-1}}} \sqrt{[\cP_{\cI_k}\Omega^2](x_{\alpha_1},\dots,x_{\alpha_q})}\; \inner{b}{\theta_{\alpha_1}\otimes \dots \otimes \theta_{\alpha_q}} + z_{\alpha}, \quad \alpha \in [n]^q
\end{align}
satisfy Condition~\ref{cond:eps_approx} (from the triangle inequality $\bmy^{(k)}$ verifies~\eqref{eq:W2eps_simple}), so that by Theorem~\ref{th:approx} the associated free energy $F_n^{(k)}$ is $\eps$-close to $F_n$ as $n \to \infty$, that is
\begin{align} \label{eq:stepping_approx}
    \limsup_{n \to \infty} \, \abs[\big]{F_n - F_n^{(k)}} \leq \eps.
\end{align}
Furthermore, under the assumption that $\theta_i \sim_\iid \bmeas$ from a compactly supported measure and that the limiting inverse-variance profile $\Omega$ is bounded, Condition~\ref{th:Flimit} can be verified via the Efron-Stein inequality (see e.g.,~\cite[Appendix C]{reeves2020}).

From now on, we write $\Omega_k^2 \coloneqq \cP_{\cI_k}\Omega^2$ and for any $(j_1, \dots, j_q) \in [k]^q$ we use $\Omega_k^2(j_1,\dots,j_q)$ to denote the constant value over $\cI_{k,j_1} \times \dots \times \cI_{k,j_q}$ of $\Omega_k^2$.
Letting $H : \psd^d \to \bbR$ be the relative entropy $\div[\big]{(S^{1/2})_\#\bmeas * \gaus}{\gaus}$, the potential function $\Psi^{(k)}: (\psd^d)^k \times (\psd^d)^k \to \bbR$ prescribed by Theorem~\ref{th:Flimit} is
\begin{align}
    \Psi^{(k)}(Q,S) &= \sum_{j=1}^k \beta_j^{(k)} H(S_j) - \frac{1}{2} \sum_{j=1}^k \beta_j^{(k)}\inner{Q_j}{S_j} \\
    & \qquad + \frac{t}{2} \sum_{j_1,\dots,j_q=1}^k \beta_{j_1}^{(k)}\cdots\beta_{j_q}^{(k)} \Omega_k^2(j_1,\dots,j_q) \; \inner{bb^\top}{Q_{j_1}\otimes \dots \otimes Q_{j_q}},
\end{align}
where $\beta_j^{(k)} \in [0,1]$ is the limiting proportion of entries $x_i$ that fall within partition element $\cI_{k,j}$ for $j=1,\dots,k$. Then, we have that
\begin{align} \label{eq:stepping_Flimit}
    \lim_{n \to \infty} F_n^{(k)} = \adjustlimits\sup_{Q \in (\psd^d)^k} \inf_{S \in (\psd^d)^k} \Psi^{(k)}(Q,S) \eqqcolon F^{(k)}.
\end{align}
Combining displays~\eqref{eq:stepping_approx} and~\eqref{eq:stepping_Flimit}, we obtain that for every $\eps>0$ there exists some $\bar k \in \bbN$ such that, for all $k,k' \geq \bar k$, $\abs{F^{(k)}-F^{(k')}} \leq \eps$. Then, $(F^{(k)})_{k \in \bbN}$ is a Cauchy sequence and, by completeness, $F^{(k)} \to F$ as $k \to \infty$ for some well-defined $F$, so that 
\begin{align}
    \lim_{n \to \infty} F_n = \lim_{k \to \infty} F^{(k)} = F.
\end{align}

\paragraph{Step 2: Monotonicity of variational problems.} Let us now focus on the potential function $\Psi^{(k)} : (\psd^d)^k \times (\psd^d)^k \to \bbR$. As a first observation, notice that the optimization of $\Psi^{(k)}$ is equivalent to optimizing a functional $\Psi_k : \cS([0,1]) \times \cS([0,1]) \to \bbR$,
\begin{align}
    \Psi_k(\sfQ,\sfS) &= \int_{[0,1]} H(\sfS(x)) \dd x + \frac{1}{2} \int_{[0,1]} \inner{\sfS(x)}{\sfQ(x)} \dd x \\
    & \quad + \frac{t}{2}\int_{[0,1]^q} \Omega_k^2(x_1,\dots,x_q) \: \inner[\big]{bb^\top}{ \sfQ(x_1) \otimes \dots \otimes \sfQ(x_q) } \dd(x_1,\dots,x_q)
\end{align}
under the constraint that $\sfQ,\sfS \in \cS([0,1])$ are constant over elements of $\cI_k$. Letting $\cS(\cI_k)$ denote this restricted space, we have that for all $k \in \bbN$
\begin{align}
    \adjustlimits\sup_{Q \in (\psd^d)^k} \inf_{S \in (\psd^d)^k} \Psi^{(k)}(Q,S) = \adjustlimits\sup_{\sfQ \in \cS(\cI_k)} \inf_{\sfS \in \cS(\cI_k)} \Psi_k(\sfQ,\sfS).
\end{align}
We observe that, for any $\sfQ \in \cS(\cI_k)$ and $x \in [0,1]$, the restriction that $\sfS \in \cS(\cI_k)$ instead of $\cS([0,1])$ is immaterial since $H$ depends on $x$ exclusively through $\sfS(x)$. As such, we have that
\begin{align}
    \inf_{\sfS \in \cS(\cI_k)} \int_{[0,1]} H(\sfS(x)) - \frac{1}{2} \inner{\sfS(x)}{\sfQ(x)} \dd x &= \inf_{\sfS \in \cS([0,1])} \int_{[0,1]} H(\sfS(x)) - \frac{1}{2} \inner{\sfS(x)}{\sfQ(x)} \dd x \\
    &= - \int_{[0,1]} H^*(\sfQ(x)/2) \dd x
\end{align}
where $H^*$ denotes the (pointwise) convex conjugate of $H$ at each $\sfQ(x)/2 \in \psd^d$. Substituting back into the potential function, we obtain 
\begin{align}
    F^{(k)} = \sup_{\sfQ \in \cS(\cI_k)} & \Big\{\frac{t}{2}\int_{[0,1]^q} \Omega_k^2(x_1,\dots,x_q)\: \inner[\big]{bb^\top}{ \sfQ(x_1) \otimes \dots \otimes \sfQ(x_q) } \dd(x_1,\dots,x_q)\\
    & \qquad - \int_{[0,1]} H^*(\sfQ(x)/2) \dd x \Big\}.
\end{align}
Define the map $W : \cS([0,1]) \to L^2([0,1]^q)$, whose output for an input $\sfQ \in \cS([0,1])$ is the function $W(\sfQ): (x_1,\dots,x_q) \mapsto \inner[]{bb^\top}{ \sfQ(x_1) \otimes \dots \otimes \sfQ(x_q) }$. 
From the multilinear structure of $W(\sfQ)$, notice that whenever $\sfQ \in \cS(\cI_k)$ the function $W(\sfQ)$ is invariant to the stepping operator over $\cI_k$. Conversely, for every $\sfQ' \in \cS([0,1])$ there exists some piecewise constant function $\sfQ \in \cS(\cI_k)$ such that $\cP_{\cI_k}[W(\sfQ')] = W(\sfQ)$. This implies that
\begin{align}
    \cP_{\cI_k}[W(\sfQ)] = W(\sfQ) \iff \sfQ \in \cS(\cI_k).
\end{align}
Since the stepping operator is an orthogonal projection in $L^2([0,1]^q)$, then, we deduce that for every $k \in \bbN$ and $\sfQ' \in \cS([0,1])$ there exists some $\sfQ \in \cS(\cI_k)$ such that
\begin{align}
    &\int_{[0,1]^q} \Omega_k^2(x_1,\dots,x_q) \cdot [W(\sfQ')](x_1,\dots,x_q) \dd(x_1,\dots,x_q) \\
    &\qquad = \int_{[0,1]^q} [\cP_{\cI_k}\Omega^2](x_1,\dots,x_q) \cdot \cP_{\cI_k}[W(\sfQ')](x_1,\dots,x_q) \dd(x_1,\dots,x_q)\\
    &\qquad = \int_{[0,1]^q} \Omega^2(x_1,\dots,x_q) \cdot [W(\sfQ)](x_1,\dots,x_q) \dd(x_1,\dots,x_q).
\end{align}
This allows us to conclude that the unique potential function $\Psi : \cS([0,1]) \times \cS([0,1]) \to \bbR$
\begin{align}
    \Psi(\sfQ,\sfS) = \int_{[0,1]} \Big( H(\sfS) - \frac{1}{2} \inner{\sfS}{\sfQ} \Big) \dd x
    + \frac{t}{2}\int_{[0,1]^q} \Omega^2 \cdot W(\sfQ) \dd(x_1,\dots,x_q)
\end{align}
allows to express each $F^{(k)}$ as the solution to the constrained optimization problem
\begin{align}
    F^{(k)} = \adjustlimits \sup_{\sfQ \in \cS(\cI_k)} \inf_{\sfS \in \cS([0,1])} \Psi(\sfQ,\sfS).
\end{align}
Since the sequence $(\cI_k)_{k \in \bbN}$ is a sequence of successively finer partitions that eventually separates points, we have that $F^{(k)} \leq F^{(k+1)}$ for all $k \in \bbN$, and by monotonicity it follows that
\begin{align} 
    F = \lim_{k \to \infty} F^{(k)} = \adjustlimits \sup_{\sfQ \in \cS([0,1])} \inf_{\sfS \in \cS([0,1])} \Psi(\sfQ,\sfS).
\end{align}
Putting the conclusions of the two proof steps together, we conclude that
\begin{align}
    \lim_{n \to \infty} F_n = \adjustlimits \sup_{\sfQ \in \cS([0,1])} \inf_{\sfS \in \cS([0,1])} \Psi(\sfQ,\sfS),
\end{align}
completing the proof of Theorem~\ref{th:hlr}.

\section{Proofs for \texorpdfstring{$q$-adic}{q-adic} assignment problem} \label{s:proof_qap}
\subsection{Proof of Proposition~\ref{prop:qap_info_equivalence}} \label{ss:proof_qap_info_equivalence}
From the chain rule of mutual information, observe that 
    \begin{align}
        I(\sigma ;\bmu,\bmv,\bmy)  & =  \underbrace{I(\sigma ;\bmu)}_{=0}+ I(\sigma ; \bmy \mid \bmu ) + \underbrace{I( \sigma ; \bmv\mid \bmy, \bmu)}_{=0} 
    \end{align}
    where $I(\sigma; \bmu) = 0$ by independence and $I(\sigma; \bmv \mid \bmy, \bmu)$ since $\bmv$ is assumed to be a deterministic function of $\bmu$. Then,
    \begin{align}
        I(\sigma ; \bmy \mid \bmu ) & = I(\sigma, \bmtheta ; \bmy \mid  \bmu)-  \underbrace{I( \bmtheta ; \bmy \mid  \bmu,\sigma)}_{=0} \\ 
        &= I(\bmtheta ; \bmy \mid  \bmu) + \underbrace{I(\sigma ; \bmy \mid  \bmtheta,\bmu)}_{=0}.
    \end{align}
    Here $I( \bmtheta ; Y \mid  \bmu,\sigma)$ because $\bmtheta = \sigma\bmu$ and $I(\sigma; \bmy \mid \bmtheta,\bmu) = 0$ because $\sigma$ is independent of $\bmy$ given $(\bmu,\bmtheta)$. Finally, we note that $I(\sigma;\bmu \mid \bmx,\bmy) = 0$, so that $I(\sigma; \bmu,\bmv,\bmy) = I(\sigma; \bmv,\bmy)$, and we conclude that $I(\sigma; \bmy,\bmv) = I(\bmtheta; \bmy \mid \bmu)$ as desired.
\subsection{Entropy with permutations}
As a first step, we present and prove a general $L^2$ bound on the difference in free energy between Gaussian convolution of two measures. While this result can be shown directly via the HWI inequality, we present an alternative proof of the statement to keep the presentation self-contained.
\begin{lem} \label{lem:moment_bound}
    Let $\mu, \nu$ be two probability measures on $\bbR^n$ with finite second moments, and let $\gaus$ denote the standard Gaussian measure on $\bbR^n$. Then, for all pairs of random variables $(u,v) \in \bbR^n \times \bbR^n$ such that $u \sim \mu$ and $v \sim \nu$,
    \begin{align}
        \abs[\big]{\div{\mu * \gaus}{\gaus} - \div{\nu * \gaus}{\gaus}} \leq \sqrt{\bbE\norm{u}^2 \vee \bbE\norm{u}^2} \, \sqrt{\bbE\norm{u - v}^2} + \frac{1}{2}\bbE\norm{u - v}^2.
    \end{align}
\end{lem}
\begin{proof}
    For any pair of random variables $(u,v) \sim \gamma \in \Gamma(\mu,\nu)$ and $t \in [0,1]$ and independent $z,z' \sim_{\iid} \normal(0,\id_n)$, we introduce the interpolating observation
    \begin{align}
        y(t) = \begin{dcases}
            \sqrt{t} u + z \\
            \sqrt{1-t} \, v + z'
        \end{dcases}
    \end{align}
    and let $H(t) : [0,1] \to \bbR$ be the relative entropy
    \begin{align}
        \ex[\bigg]{\log \!\bigg(\int{ \exp\!\Big\{ \sqrt{t} \inner{u + z}{u'} - \frac{t}{2} \norm{u'}^2 + \sqrt{1-t} \inner{v + z'}{v'} - \frac{1-t}{2} \norm{v'}^2 \Big\} \, \gamma(\dd u', \dd v') \bigg)}},
    \end{align}
    where the expectation is taken with respect to all random variables $(u,v,z,z')$. From the I-MMSE relation and the fundamental theorem of calculus,
    \begin{align}
        \div{\mu * \gaus}{\gaus} - \div{\nu * \gaus}{\gaus} &= \frac{1}{2} \int_0^1 H'(t) \dd t \\
        &= \frac{1}{2} \int_0^1 \ex*{\norm[\big]{\ex{u \mid y(t)}}^2} - \ex*[\big]{\norm[\big]{\ex{v \mid y(t)}}^2} \dd t.
    \end{align}
    Using the definitions $\hat u_t \coloneqq \ex{u \mid y(t)}$ and $\hat v_t \coloneqq \ex{v \mid y(t)}$, then, an application of Cauchy-Schwarz and Jensen's inequalities yield that
    \begin{align}
        \bbE\norm{\hat u_t}^2  - \bbE \norm{\hat v_t}^2 &= 2\ex{\inner{\hat u_t}{\hat u_t - \hat v_t}} - \ex{\norm{\hat u_t - \hat v_t}^2} \\
        &\le 2\sqrt{\ex{\norm{\hat u_t}^2}} \sqrt{\ex{\norm{\hat u_t - \hat v_t}^2}} +  \ex{\norm{\hat u_t - \hat v_t}^2}\\
        &\le   2\sqrt{ \ex{ \| u\|^2}} \sqrt{ \ex{ \| u -  v \|^2 }} + \ex{ \| u -  v\|^2} \\
        &= 2 \sqrt{\bbE\norm{u}^2}\sqrt{\bbE\norm{u-v}^2} + \bbE\norm{u-v}^2.
    \end{align}
    From a symmetric argument swapping the roles of $u$ and $v$, then, we conclude that
    \begin{align}
        \abs[\big]{\div{\mu * \gaus}{\gaus} - \div{\mu * \gaus}{\gaus}} &\leq \frac{1}{2} \int_0^1 2\sqrt{\bbE\norm{u}^2 \vee \bbE\norm{u}^2} \, \sqrt{\bbE\norm{u - v}^2} + \frac{1}{2}\bbE\norm{u - v}^2 \dd t \\
        &= \sqrt{\bbE\norm{u}^2 \vee \bbE\norm{u}^2} \, \sqrt{\bbE\norm{u - v}^2} + \frac{1}{2}\bbE\norm{u - v}^2.
    \end{align}
\end{proof}

Let $(u_i)_{i \in \bbN}$ be a sequence of points in $\bbR^d$. For $n \in \bbN$, define $\bmu \coloneqq (u_1,\dots,u_n)$, let $\sigma_n \sim \unif(\sym([n]))$ be a random permutation of $[n]$ drawn at uniformly at random from $\sym([n])$, and define $\bmtheta \coloneqq \sigma_n \bmu \coloneqq (u_{\sigma_n(1)}, \dots, u_{\sigma_n(n)})$. We denote the induced distribution on $\bmtheta$ by $\bmeas^n$ and  the empirical measure of $\bmu$ (equivalently $\bmtheta$) by $\pi_n \coloneqq \frac{1}{n}\sum_{i=1}^n \delta_{u_i}$.
Here, we will argue that the relative entropy $H_n : \psd^d \to \bbR$ defined as
\begin{align} \label{eq:permute_entropy}
    H_n(S) \coloneqq \frac{1}{n} \div[\big]{(S^{1/2} \otimes \id_n)_\# \bmeas^n * \gaus}{\gaus}
\end{align}
can be characterized in terms of the limiting behavior of the empirical measure $\pi_n$ as the sample size grows. 
\begin{lem} \label{lem:entropy_permutation_limit}
    For each $n \in \bbN$, let $H_n$ be defined according to~\eqref{eq:permute_entropy}. Assume $\pi_n$ converges in quadratic Wasserstein distance to some probability measure $\bmeas$ on $\bbR^d$ with bounded second moment, i.e. $W_2(\pi_n,\bmeas) \to 0$ as $n \to \infty$. Define the function $H(S) : \psd^d \to \bbR$
    \begin{align}
        H(S) \coloneqq \div[\big]{(S^{1/2})_\# \bmeas * \gaus}{\gaus}.
    \end{align}
    Then, for all $S \in \psd^d$, the sequence $(H_n)_{n \in \bbN}$ satisfies
    \begin{align}
        \abs{H_n(S) - H(S)} \to 0
    \end{align}
    pointwise as $n \to \infty$.
\end{lem}
\begin{proof}
    Let $v_i \sim_\iid \bmeas$, and for $n \in \bbN$ define $\bmv \coloneqq (v_1,\dots,v_n) \sim \bmeas^{\otimes n}$. Then, for any $n \in \bbN$, it holds that
    \begin{align}
        \frac{1}{n}\div[\big]{(S^{1/2} \otimes \id_n)_\#\bmeas^{\otimes n} * \gaus}{\gaus} = \frac{1}{n} \sum_{i=1}^n \div[\big]{(S^{1/2})_\#\bmeas * \gaus}{\gaus} = H(S)
    \end{align}
    for all $S \in \psd^d$, so that
    \begin{align}
        \frac{1}{n} \abs[\big]{\div[\big]{(S^{1/2} \otimes \id_n)_\#\bmeas^n * \gaus}{\gaus} - \div[\big]{(S^{1/2} \otimes \id_n)_\#\bmeas^{\otimes n} * \gaus}{\gaus}} = \abs[\big]{H_n(S) - H(S)}.
    \end{align}
    Letting $\bmtheta \sim \bmeas^n$, we see by Lemma~\ref{lem:moment_bound} that for any joint distribution of $(\bmtheta,\bmv) \sim \gamma \in \Gamma(\bmeas^n, \bmeas^{\otimes n})$ it holds
    \begin{align}\label{eq:em_entropy_bound}
        \abs[\big]{H_n(S) - H(S)}
        \leq \frac{1}{n} \norm{S}_{\op} ^2\Big( \sqrt{ \bbE\| \bmv  \|^2\vee\bbE\| \bmw  \|^2} \sqrt{\bbE \|\bmv  - \bmw\|^2} + \frac{1}{2}\bbE \|\bmv  - \bmw \|^2 \Big).
    \end{align}
    Clearly, $\frac{1}{n}\max\set{\bbE\| \bmtheta  \|^2,\bbE\| \bmv  \|^2}$ is bounded since $\bmeas$ has bounded second moment. Furthermore, we consider an explicit coupling of $(\bmtheta,\bmv)$ by permuting the entries of $\bmv$ by some $\sigma^* \in \sym([n])$ such that
    \begin{align}
        \norm{\bmv - \sigma^*\bmw}^2 = \min_{\sigma \in \sym([n])}\norm{\bmv -\sigma \bmw}^2.
    \end{align}
    Conditionally on $\bmv$, then, $\sigma^*$ constitutes the optimal coupling between the empirical measures $\pi_n$ of $\bmtheta$ and and $\nu_n$ of $\bmv$, that is
    \begin{align}
        \frac{1}{n} \norm{\bmtheta - \sigma^*\bmv}^2 = W_2^2(\pi_n,\nu_n),
    \end{align}
    for all $n \in \bbN$. 
    Since $W_2^2(\pi_n, \bmeas) \to 0$ by assumption and $\ex{W_2^2(\nu_n, \bmeas)} \to 0$ (see e.g.~\cite[Theorem 1]{fournier2015}), it follows that the right hand side of~\eqref{eq:em_entropy_bound} converges to zero as $n \to \infty$ pointwise for all $S \in \psd^d$.
\end{proof}

\subsection{Proof of Theorem~\ref{th:qap}}
As a preliminary remark, we note that for the $q$-adic assignment problem under Condition~\ref{as:qap} satisfies the free energy concentration required by Condition~\ref{as:Flimit}. This can be shown by noting that the noise component of the observation model is invariant to permutation of the indices and the random permutation is independent of everything else in the model. Therefore, free energy concentration can be shown by concentration of the free energy for a decoupled model in which $\bmtheta$ is drawn i.i.d. from $\bmeas$, which we recall is taken to have compact support.

To start the proof, Lemma~\ref{lem:entropy_permutation_limit} immediately gives that, for every $d \in \bbN$, the pointwise limit of the scalar channel relative entropy $H_n$ is given by the single-letter relative entropy $H : \psd^d \to \bbR$,
\begin{align}
    H(S) \coloneqq \div[\big]{(S^{1/2})_\#\bmeas * \gaus}{\gaus},
\end{align}
which is continuously differentiable (and in fact infinitely differentiable, see e.g.~\cite{guo2011}). Fix any two $\eps,\eps' > 0$. Under Condition~\ref{as:qap}, we can combine Theorems~\ref{th:approx} and~\ref{th:Flimit} to conclude that there exist $d_\eps,d_{\eps'} \in \bbN$ depending on $\eps$ and $\eps'$, respectively, such that
\begin{align}
    \limsup_{n \to \infty} \;\abs{F_n - F^{(\eps)}} \leq \eps \quad \text{and} \quad \limsup_{n \to \infty} \;\abs{F_n - F^{(\eps')}} \leq \eps' .
\end{align}
Here we defined $F^{(\eps)}$ and $F^{(\eps')}$ to be the asymptotic free energies of approximating models $\bmy^{(\eps)}$ and $\bmy^{(\eps')}$ defined as in~\eqref{eq:y_eps}. This implies that for each $\delta > 0$ there exist $\bar\eps$ such that for all $0 < \eps,\eps' < \bar \eps$ it holds
\begin{align}
    \abs{F^{(\eps)} - F^{(\eps')}} \leq \delta.
\end{align}
This implies that $(F^{(\eps)})_{\eps > 0}$ is a Cauchy net, so that by completeness $F^{(\eps)} \to F$ for some value $F$ as $\eps \to 0$. Since $(F^{(\eps)})$ admits a limiting value $F$, then, it must be the case that
\begin{align}
    \lim_{n \to \infty} \abs{F_n - F} = 0,
\end{align}
from which Theorem~\ref{th:qap} follows.

\section{Proof of auxiliary results} \label{ap:approx_lemmas}
\subsection{Proof of Lemma~\ref{lem:F_invariance}}
We first state and prove the following orthogonality result.
\begin{lem} \label{lem:entropy_proj}
Let $\bmeta =[\eta_1, \dots, \eta_n]$ be a $d \times n$ matrix and $\bmz$ be a $d \times n$ standard Gaussian matrix independent of $\bmeta$. For $S \in \psd^d$, Define
\begin{align}
    I(S) \coloneqq I(\bmeta ; S^{1/2} \bmeta + \bmz ).
\end{align}
Let  $U_1, \dots, U_k$ be a collection of $d \times d$ orthogonal projection matrices satisfying  $U_jU_{j'}= 0$ for all $j \neq j' \in [k]$. And define a projector $\cP : \bbR^{d \times d} \to \bbR^{d \times d}$ as
\begin{align}
    \cP(M) = \sum_{j=1}^k U_j M U_j.
\end{align}
Then, if there exists mapping $\psi \colon [n] \to [k]$ such that $U_{\psi(i)} \theta_i = \theta_i$ almost surely for $i = 1, \dots, n$,
\begin{align}
    I(S) = I(\cP(S)) \quad \text{for all $S \in \psd^d$}. 
\end{align}
\end{lem}
\begin{proof}
    For $S, S'\in \psd^d$ the I-MMSE relation yields 
\begin{align}
I(S) - I(S') = \sum_{n=1}^N \int_0^1 \inner[\big]{ \ex{ \cov{\eta_n \mid S_{\lambda}^{1/2}\bmeta + \bmz }}}{S-S'} \dd\lambda
\end{align}
where $S_\lambda \coloneqq \lambda S + (1- \lambda) S'$. 
Evaluating with $S' = \cP(S)$, we see that 
\begin{align}
I(S) - I(\cP(S)) & =\sum_{i=1}^n \int_0^1 \inner[\big]{ \ex{ \cov{\eta_i \mid S_{\lambda}^{1/2}\bmeta  + \bmz }}}{S-\cP(S)} \dd\lambda \\
& =\sum_{i=1}^n \int_0^1 \langle U_{\psi(i)} \ex{  \cov{\eta_i \mid S_\lambda^{1/2}\bmeta  + \bmz }} U_{\psi(i)},  S - \cP(S)  \rangle \dd\lambda \\
& =\sum_{i=1}^n \int_0^1 \langle \ex{  \cov{\eta_i \mid S_{\lambda}^{1/2}\bmeta  + \bmz }} ,  U_{\psi(i)} (S-\cP(S))U_{\psi(n)}  \rangle \dd\lambda \\
&= 0
\end{align}
where the last step holds because $U_j S U_j = U_j \cP(S) U_j$ for all $k = 1, \dots, k$. 
\end{proof}

\paragraph{Proof of Lemma~\ref{lem:F_invariance}.} 
For the proof, we begin by noting that, from orthogonal invariance of Gaussian distributions (see e.g.~\cite{reeves2018}), the model~\eqref{eq:y_lift} is statistically equivalent to the symmetrized observations 
\begin{align}
    y_\alpha' = \sqrt{\frac{t}{n^{q-1}}} (\widehat{A}^\top\widehat{A})^{1/2} \: (\widehat \eta_{\alpha_1} \otimes \dots \widehat \eta_{\alpha_q}) + z_\alpha', \quad \alpha \in [n]^q, 
\end{align}
for $z_\alpha' \sim_\iid \normal(0,\id_{dk})$ independent of everything else.
By observing that the matrices $U_j$ are orthogonal projections such that $U_\kappa U_{\kappa'} = 0$ for all $\kappa \neq \kappa' \in [k]^q$ and furthermore $U_{\psi(x_{\alpha_1})\dots\psi(x_{\alpha_q})} (\widehat \eta_{\alpha_1} \otimes \dots \widehat \eta_{\alpha_q}) = (\widehat \eta_{\alpha_1} \otimes \dots \widehat \eta_{\alpha_q})$, we can appeal to Lemma~\ref{lem:entropy_proj} above as well as identity~\eqref{eq:I_F_identity} to deduce the desired result.

\subsection{Proof of Lemma~\ref{lem:potential_diagonal}}
For any $\widehat Q, \widehat S \in \psd^{dk}$, by definition of the trace inner product and Lemma~\ref{lem:entropy_proj} above we have
\begin{align}
    \widehat H(\cP_1(\widehat{S})) - \frac{1}{2} \inner{\cP_1(\widehat{Q})}{\cP_1(\widehat{S})} = H(S) - \frac{1}{2}\sum_{j=1}^k \inner{Q_j}{S_j},
\end{align}
where we recall that $H : (\psd^d)^k \to \bbR$ is the pointwise limit of the linear channel relative entropy $H_n$ we defined in~\eqref{eq:entropy_linear}.
For any $q \in \bbN$, then, we observe that from the identity~\eqref{eq:kron_commute} there exists an orthogonal matrix $\Pi_q$ such that
\begin{align}
    \cP_1(\widehat{Q})^{\otimes q} &= \Big( \sum_{k=1}^K U_k \widehat{Q} U_k \Big)^{\otimes q} \\
    &= \Big( \sum_{j=1}^k e_je_j^\top  \otimes Q_j \Big)^{\otimes q} \\
    &= \sum_{\kappa \in [k]^q} \left( \Pi_q^\top \left( e_{(\kappa)}e_{(\kappa)}^\top \otimes Q^{\otimes \kappa} \right)\Pi_q \right) 
\end{align}
where we defined $e_{(\kappa)} \coloneqq e_{\kappa_1} \otimes \cdots \otimes e_{\kappa_p}$ for all $\kappa \in [k]^q$. Then, from orthogonality of the bases $e_{(\kappa)}$ and the definition of $\widehat{A}$ it follows that
\begin{align}
    \inner{\cP_q(\widehat{A}^\top \widehat{A})}{\cP_1(\widehat{Q})^{\otimes q}} &= \inner[\Big]{\sum_{\kappa \in[k]^p} \Pi_q^\top \big( e_{(\kappa)}e_{(\kappa)}^\top \otimes A_\kappa^\top A_\kappa \big) \Pi_q }{\sum_{\kappa \in [k]^q} \Pi_q^\top \big( e_{(\kappa)}e_{(\kappa)}^\top  \otimes Q^{\otimes \kappa} \big) \Pi_q} \\
    &= \sum_{\kappa \in [k]^q} \inner[\Big]{A_\kappa^\top A_\kappa}{Q^{\otimes \kappa}}.
\end{align}
Therefore, we conclude that
\begin{align}
\Psi^\cP (\cP_1(\widehat{Q}), \cP_1(\widehat{S})) = H(S) - \frac{1}{2} \sum_{j=1}^k \langle S_j , Q_j \rangle + \frac{1}{2} \sum_{\kappa \in [k]^{q}} \inner[\Big]{A_\kappa^\top A_\kappa}  {Q^{\otimes \kappa}} = \Psi(Q,S).
\end{align}

\end{document}